\documentclass[11pt]{article}
\usepackage{graphicx}
\usepackage{psfrag}
\usepackage{varwidth}
\usepackage{amsmath,amssymb,amsfonts,amsthm}
\usepackage{booktabs}
\usepackage{natbib}
\usepackage{comment}
\usepackage{color}
\def \IR{\hbox{{\rm I}\kern-.2em\hbox{{\rm R}}}}
\newcommand{\mv}[1]{{\boldsymbol{\mathrm{#1}}}}
\newcommand{\proper}{\mathsf}
\newcommand{\pV}{\proper{V}}
\DeclareMathOperator{\diag}{diag}
\DeclareMathOperator{\R}{\mathbb{R}}
\newcommand{\scal}[2]{\left\langle {#1},\,{#2} \right\rangle}
\usepackage{amsthm}

\usepackage{algorithm}
\usepackage{algpseudocode}

 \usepackage[normalem]{ulem}
 \definecolor{darkgreen}{rgb}{0, 0.8, 0.7}

\newtheorem*{prop*}{Proposition}

\newcommand{\E}{{ {\rm E}}}
 \graphicspath{{figs/}{./}}
 \theoremstyle{remark}

\usepackage{color}

\title{Linear Mixed-Effects Models for Non-Gaussian Repeated Measurement Data}

\author{{\"O}zg{\"u}r Asar$^1$, David Bolin$^{2}$, Peter J. Diggle$^3$, Jonas Wallin$^{4, *}$}
\date{%
    $^1$Department of Biostatistics and Medical Informatics, \\ Ac{\i}badem Mehmet Ali Ayd{\i}nlar University, {\.I}stanbul, Turkey\\%
    $^2$Mathematical Sciences, Chalmers University of Technology and the University of Gothenburg, Gothenburg, Sweden\\
		$^3$CHICAS, Lancaster Medical School, Lancaster University, Lancaster, United Kingdom\\
		$^4$Department of Statistics, Lund University, Lund, Sweden \\[2ex]
		$^*$ jonas.wallin@stat.lu.se
}

\begin{document}

\maketitle

{\bf Abstract}

 We consider the analysis of continuous repeated measurement outcomes 
 that are collected through time, also known as longitudinal data. 
 A standard framework for analysing data of this kind is a linear Gaussian mixed-effects model
within which the outcome variable can be decomposed into fixed-effects, 
 time-invariant and time-varying random-effects, 
 and measurement noise. We  develop methodology that, for the first time,
allows any combination of
these stochastic components to be non-Gaussian, using multivariate 
 Normal variance-mean mixtures. We estimate parameters by maximum likelihood,
 implemented with a novel, computationally efficient stochastic gradient algorithm. 
We obtain standard error estimates 
 by inverting the observed Fisher-information matrix, and obtain the predictive distributions
for the random-effects in both filtering (conditioning on past and current data) and smoothing (conditioning on all data)
contexts.  To implement these
 procedures, we introduce an {\tt R} package, {\tt ngme}. We 
 re-analyse two data-sets, from cystic fibrosis and nephrology research,
 that were previously analysed using Gaussian linear mixed effects models.


 \section{Introduction}

 This paper is concerned with the analysis of real-valued {\it repeated measurement} data 
 that are collected through time, also known as {\it longitudinal} data. 
 The basic data-structure is that repeated measurements of an {\it outcome variable} are made on each of a number 
 of {\it subjects} at each of a number of {\it follow-up times}, not necessarily 
 the same for all subjects, with explanatory variables or {\it covariates} of 
 two kinds also available: {\it baseline} covariates attached to subjects; 
 and {\it longitudinal} covariates attached to individual outcomes. 
 We write $Y_{ij}$ for the $j$th measurement of the outcome on the $i$th subject, 
 $t_{ij}$ for the  corresponding follow-up time, $\mv{a}_i$ for the vector of 
 baseline covariates associated with the $i$th subject and $\mv{x}_{ij}$ 
 for the vector of longitudinal covariates attached to the
 $j$th measurement on the $i$th subject.  
 
 Figure \ref{fig:PANSS} shows a simple example, taken from a randomised trial 
 of drug treatments for schizophrenia, in which the outcome variable is a measure of each 
 subject's mental state at times 0, 1, 2, 4, 6 and 8 weeks after randomisation 
 to one of two different drug therapies, placebo vs. active treatment. Here, $a_i$ is 
 a scalar treatment indicator, whilst the general pattern of decreasing responses over time 
 suggests a quadratic trend, hence $\mv{x}_{ij} = [1 \ t_{ij} \ t_{ij}^2]^{\top}$. 
 The figure shows data from three subjects in each of the two treatment arms; 
 the complete trial included 88 subjects in the placebo group and 
 435 subjects distributed over five active  treatment arms \citep{henderson2000}. This example shows 
 several features that are typical of studies of this kind: 
 the outcome variable, PANSS (Positive and Negative Syndrome Scale;
Kay et al, 1987), \nocite{kay1987}
 is an imperfect measurement instrument for the underlying process of interest, namely each subject's
state of mental health at the time of measurement; the outcome variable
 exhibits stochastic variation both between subjects and between follow-up times within subjects; 
 questions of interest include {\it estimation} of parameters that define 
 the mean response profiles of the underlying process over time and 
 {\it prediction} of the trajectory of the process for an individual subject.
  
      Most of the very extensive literature on statistical methods for data of this kind uses 
	either a Gaussian model or, if the inferential goal is restricted to parameter estimation, a set of
	estimating equations; text-book accounts include \cite{verbeke2001}, 
      \cite{diggle2002} and \cite{fitzmaurice2011}.
	In this paper, we present methodology for handling repeated measurement data that exhibit 
	long-tailed or skewed departure from Gaussian distributional assumptions. In Section \ref{sec:literature}
      we  review the literature 
     on existing approaches to  Gaussian and non-Gaussian modelling of real-valued repeated measurement
     data. In Section \ref{sec:model}, 
	we set out our proposed class of non-Gaussian models. 
       In Section \ref{sec:inference}, we describe a computationally 
	fast method for likelihood-based inference. Section \ref{sec:examples} describes two applications. In the first
	of these the scientific focus is on estimation of mean response profiles, whilst in the second
	the focus is on real-time individual-level prediction. Section \ref{sec:package} describes 
	our {\tt R} package, {\tt ngme}, that implements the new methodology. In Section \ref{sec:discussion}, 
	we discuss some potential extensions, including models for categorical or count 
  data \citep{molenberghs+verbeke2005} and joint modelling of repeated measurement and time-to-event data \citep{rizopoulos2012}.
  
 \begin{figure}[t]
 \centering
 \includegraphics[width = 0.45\linewidth]{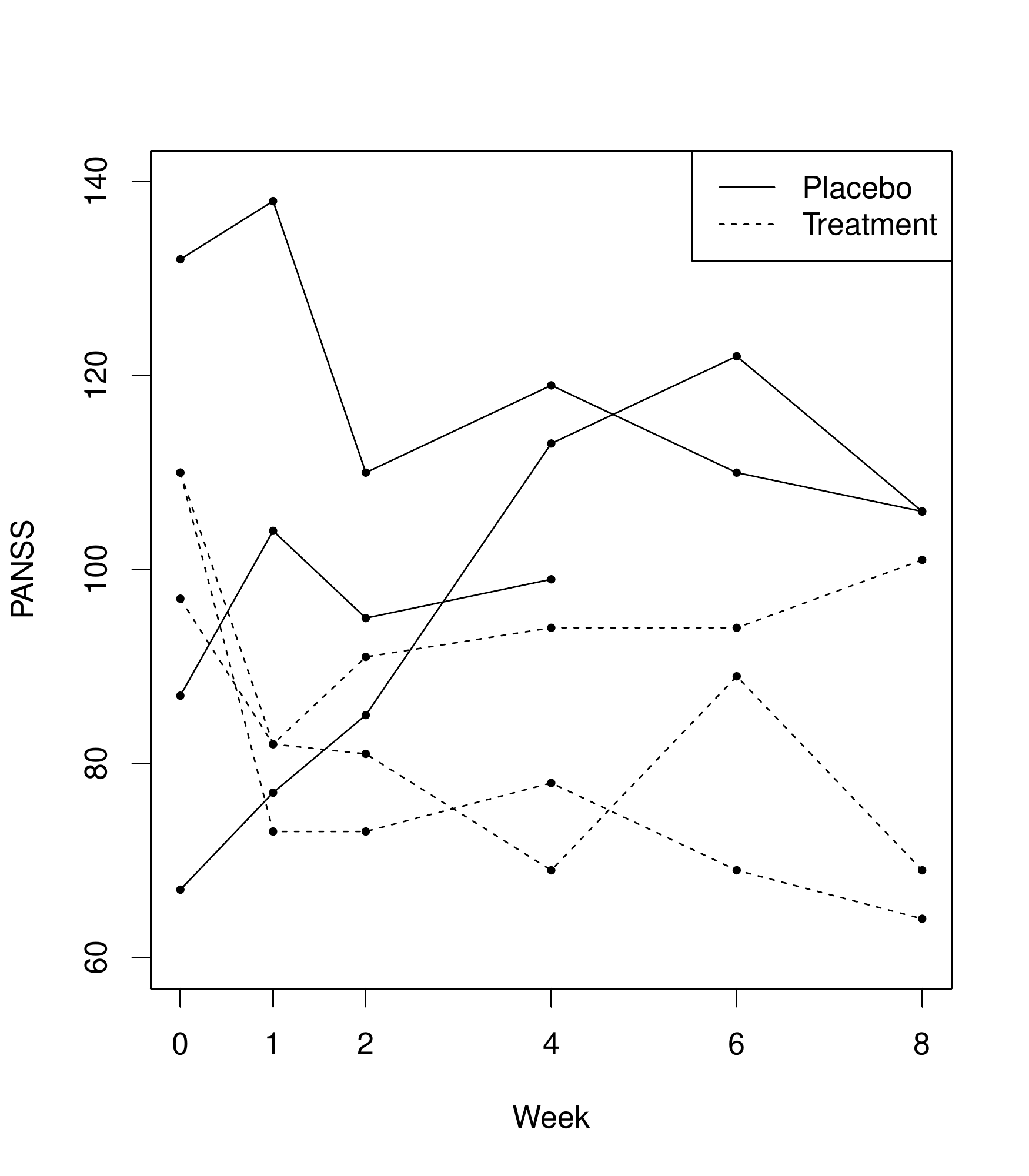}
 \caption{Data on six patients in a randomised trial of drug therapies for  schizophrenia. 
 The outcome variable, PANSS (Positive and Negative Syndrome Scale), is a questionnaire-based 
 instrument for assessing the severity of  a  patient's condition. Solid and dashed lines 
 identify patients from two different treatment arms, placebo and active treatment, respectively.}
 \label{fig:PANSS}
 \end{figure}
   
 \section{Literature review} 
 \label{sec:literature}

  \subsection{Gaussian models for real-valued repeated measurement data}  
	\label{sec:gaussian}

     \citet{laird1982} were the first authors to consider modelling repeated measurements as noisy versions of 
	underlying signals that can be decomposed into fixed effects,  
	$\mv{a}_i^{\top} \mv{\alpha} + \mv{x}_{ij}^{\top} \mv{\beta}$, 
	and random effects, $\mv{d}_{ij}^{\top} \mv{U}_i$,
     leading to the mixed-effects model
     \begin{equation}
     Y_{ij} =  \mv{a}_i^{\top} \mv{\alpha} + \mv{x}_{ij}^{\top} \mv{\beta} + 
		\mv{d}_{ij}^{\top} \mv{U}_i + \sigma Z_{ij}; \quad j = 1, \ldots, n_i, \quad i = 1, \ldots, m, 
    \label{eq:laird-ware}
    \end{equation}
     where $n_i$ is the number of measurements on 
	the $i$th subject, $m$ is the number of subjects, the individual-level
	$\mv{U}_i$ are mutually independent, zero-mean multivariate Normally distributed,
     $\mv{U}_i \sim {\rm MVN}(\mv{0}, \mv{\Sigma})$, and the $Z_{ij}$
    are mutually independent ${\rm N}(0,1)$. 

	A widely-used special case of (\ref{eq:laird-ware}) is the ``random-intercept-and-slope'' 
  model in which each subject's random effect is a linear function of time.
	This model  is very useful when the data contain only a small number
  of repeated measurements per individual. With longer sequences,  the assumption that individual random
	effect trajectories can be approximated by straight lines becomes implausible. 
	\cite{diggle1988}  proposed adding to the model a {\it time-varying} random-effect term, 
	$W_i(t_{ij})$, specified as a stationary stochastic process. \cite{taylor1994} and \cite{diggle2015} 
	later considered non-stationary 
	options for $W_i(t_{ij})$. The general specification for
	models of this kind is that
  \begin{equation}
  Y_{ij} =  \mv{a}_i^\top \mv{\alpha} + \mv{x}_{ij}^\top \mv{\beta} + \mv{d}_{ij}^\top \mv{U}_i + W_i(t_{ij}) + \sigma Z_{ij}
  \label{eq:linearGaussian}
  \end{equation}
  where, 
	in addition to the notation already introduced, the $W_i(t)$ are 
	independent copies of a continuous-time  Gaussian process with mean zero and 
	covariance function $\gamma(t, t^{\prime})= {\rm Cov}\{W_i(t), W_i(t^{\prime})\}$. 
  We consider the elements of 
	both the $\mv{a}_i$ and the $\mv{x}_{ij}$ to be pre-specified constants. 
	This implicitly assumes, in particular, that if any time-varying covariate 
	is not pre-specified, it is stochastically independent of all other terms 
	in the model, hence conditioning on it is innocuous. We can then drop the 
	term $\mv{a}_i^\top \alpha$ in (\ref{eq:linearGaussian}) by allowing 
	elements of $\mv{x}_{ij}$ to take identical values for all $j$ associated with any fixed $i$.
  For the covariance function $\gamma(t, t^{\prime})$, we use the stationary Mat\'{e}rn (1960) family,
  \begin{equation}
  \gamma(t, t^{\prime}) = \omega^2 \left\{ 2^{\phi-1} \Gamma \left( \phi \right) \right \}^{-1} 
               \left(|t - t^{\prime}| /\kappa \right)^{\phi} K_{\phi} \left(|t - t^{\prime}| /\kappa \right), 
    \label{eq:matern}
   \end{equation}
	where $\Gamma(\phi)$
  is the complete gamma function, $\phi$ is a shape parameter, 
   $\kappa$ is a scale parameter measured in units of time
	and $K_{\phi}$ is a modified Bessel function of order $\phi$. 
	The corresponding Gaussian process $W_i(t)$ is 
	$r$ times mean-square differentiable if $\phi>r$.
	An alternative way of capturing non-linear behaviour of 
	repeated measurements is to specify the random effects as regression splines 
	or polynomials with stochastic coefficients \citep[Chapter 19]{fitzmaurice2011}. 
	We do not consider these approaches in this paper, 
	since they appear to us less natural  than the stochastic process approach based on (\ref{eq:matern}) and would require many more parameters to
      achieve the same flexibility in shape.
 
  Likelihood-based inference  for the model (\ref{eq:linearGaussian}) is straightforward. 
	The likelihood function is a product of $m$ multivariate Normal densities with dimensions $n_i$. 
For typical study designs the $n_i$ are small enough that the required matrix calculations 
	are not computationally demanding. 
  
  In the continuous-time setting, it is helpful to exploit an alternative representation of 
	$W(\cdot)$ as the solution to a stochastic differential equation,
  \begin{equation}\label{eq:GaussianSDE}
  {\cal D}{W(t)} = Z(t),
  \end{equation}
  where ${\cal D}$ is a differential operator and $Z(t)$ is continuous-time 
	white noise \citep{lindgren2011}.
	For example, the integrated random walk model used 
	by \cite{diggle2015} corresponds to ${\cal D}= \frac{\partial^2}{\partial t^2} $, 
	whilst the Mat\'{e}rn model corresponds to 
	$${\cal D}= \left(\frac{\partial^2}{\partial t^2} - \kappa \right)^{(2 \phi+1)/4}.$$
	 
  In applications where only the regression parameters $\beta$ are of scientific interest, 
 estimating equations offer an alternative to likelihood-based estimation. In the current 
 context, this approach was introduced by \cite{liang1986}, working in the wider setting of 
 generalised  linear models. For linear models, the approach consists of estimating $\beta$ 
 by weighted least squares, hence 
 \begin{equation}
 \tilde{\mv{\beta}} = \left(\sum_{i=1}^m \mv{x}_i^\top \mv{F}_i \mv{x}_i \right)^{-1}
                       \left(\sum_{i=1}^m \mv{x}_i^\top \mv{F}_i \mv{Y}_i \right)
 \label{eq:ee}
 \end{equation}
 where, for each $i$, $\mv{Y}_i=(Y_{i1}, \ldots, Y_{in_i})^{\top}$, 
 $\mv{x}_i$ is the $n_i$ by $k$ matrix whose $j$th row is $\mv{x}_{ij}^{\top}$ and 
 the $\mv{F}_i$ are weight matrices.  
 Re-writing (\ref{eq:ee}) in an obvious shorthand notation as 
 $\tilde{\mv{\beta}} = \mv{D} \mv{Y}$, inference for $\mv{\beta}$ uses the result that 
 $\tilde{\beta}$ is asymptotically multivariate Gaussian with mean 
 $\mv{\beta}$ and variance $\mv{D} \mv{C} \mv{D}^\top$, where $\mv{C} = {\rm Var}(\mv{Y})$, 
 a block-diagonal matrix with non-zero blocks, $\mv{C}_i ={\rm Var}(\mv{Y}_i)$. 
 If $\mv{F}_i = \mv{C}_i^{-1}$, then $\tilde{\mv{\beta}}$ is the maximum 
 likelihood (ML) estimator for $\mv{\beta}$.

 The basic idea behind (\ref{eq:ee}) is to choose, rather than estimate, 
 a set of matrices $\mv{F}_i$ that reflect a reasonable {\it working covariance structure} 
 for the matrices $\mv{C}_i = {\rm Var}(\mv{Y}_i)$, but not to rely on the correctness of the chosen structure. 
 Instead, the unknown matrix $\mv{C}_i$ is replaced by a non-parametric estimate, $\tilde{\mv{C}}_i$. 
 One such set of estimates is given by
$\tilde{\mv{C}}_i = n_i^{-1} (\mv{Y}_i - \mv{x}_i \tilde{\mv{\beta}})(\mv{Y}_i - \mv{x}_i \tilde{\mv{\beta}})^\top.$
Individually, each $\tilde{\mv{C}}_i$ is a very poor estimate of $\mv{C}_i$, but the  
implicit averaging in (\ref{eq:ee}) leads to consistent estimation
of ${\rm Var}(\tilde{\mv{\beta}})$ in the limit $m \rightarrow \infty$ for fixed $n_i$.

\subsection{Non-Gaussian models for real-valued repeated measurement data}  
\label{sec:literature_non_Gaussian}

 The existing literature on non-Gaussian 
 models takes as its starting point a 
 linear model with correlated errors,
 \begin{equation}
 \label{eq:lin_mod_cor_error}
 Y_{ij} =  \mv{x}_{ij}^\top \mv{\beta} + S_{ij}
 \end{equation}
 where, in the case of a common set of follow-up times $t_1,...,t_n$ for each subject, the
 $\mv{S}_{i} = (S_{i1}, \ldots, S_{in})^T$ are independent copies of a zero-mean multivariate Normal random
 vector \citep{jennrich1986}.  Most authors  only consider the  Laird-Ware approach as presented in 
 \eqref{eq:laird-ware}, where  
 \begin{equation}
 S_{ij} = \mv{d}_{ij}^\top \mv{U}_i + Z_{ij}.
 \label{eq:singleterm}
 \end{equation}

\cite{liu1995}, \citet{lange1989} and  \citet{pinheiro2001} replaced each  $S_{ij}$
in (\ref{eq:lin_mod_cor_error}) or  \eqref{eq:singleterm} by
 $S^*_{ij} = S_{ij}/\sqrt{V_{i}}$ where the $V_{i}$ are mutually independent  
 unit-mean gamma-distributed random variables. 
 They estimated the model parameters by maximum
 likelihood using an EM algorithm \citep{dempster1977}.  
 \citet{lin2011} considered Bayesian methods of inference for the same class of models. 
 \citet{matos2013} extended the work of \citet{pinheiro2001} to 
 allow censored outcomes. 
 
 \citet{song2007} and \citet{zhang2009} considered an extension to \citet{lange1989} by
allowing the gamma-distributed scaling factor $V_i$ to apply to either one of the two components 
on the right hand side of (\ref{eq:singleterm}).
\citet{lin2007} apply the gamma-distributed scaling factor only to the random effect term $\mv{d}_{ij}^\top \mv{U}_i$,
but also replace the mutually independent $Z_{ij}$ by a set of autoregressive processes; this
restricts its applicability 
 to data with  equally-spaced measurement times.

  \citet{rosa2003} and \citet{tian2008} 
  also used the formulation
  $S^*_{ij} = S_{ij}/\sqrt{V_{i}}$ but  without restricting the
  $V_{i}$ to be gamma-distributed.  \citet{lange1995} called the resulting  family of
 distributions the
  {\it Normal/independent} family.  The
	{\tt R} package {\tt heavy} \citep{osorio2016} fits this class of models.  In a series of papers, V.H. Lachos
    and colleagues have developed methodology for fitting non-linear 
     mixed models using  the Normal/independent family; see 
	\citet{lachos2009}, \citet{lachos2010}, \citet{zeller2010}, \citet{lachos2011}, \citet{cabral2012},
     \citet{lachos2012} and
       \citet{lachos2013}, also independent contributions by \citet{verbeke1996}, 
     \citet{sun2008}, \citet{ho2010},  \citet{delacruz2014}, \citet{zhang2015}
      and \citet{yavuz2016}.   

 Several authors have
 extended the single-term modelling framework (\ref{eq:singleterm})
 by de-coupling the 
 scalings of the random effect term and
 the measurement error term.  See, for example, \citet{rosa2004},  \citet{aralleno_valle2007},
 \citet{jara2008},
 \citet{meza2012},
 \citet{choudhary2014} and \citet{bai2016}. 
  \citet{lu2014} extended the approach to include consideration of non-ignorable drop-out. 
 
 \citet{wang2011},
 \citet{wang2012},
 \citet{lin2013} and
 \citet{kazemi2013} used the Normal/independent  family to model multivariate repeated measurement data. 

 Other authors have taken a semi-parametric approach to the problem, for example by using
a Dirichlet process prior for the random effects. See
\citet{kleinman1998}, \citet{ghidey2004}, \citet{tao2004},
 \citet{subtil2010},
 \citet{davidian1993}, 
 \citet{zhang2001} and  
 \citet{vock2012}. 
 \citet{robustlmm} has considered robust estimating equations. 

We have  found only two papers that apply the single-term formulation (\ref{eq:singleterm})  to the
general form of the model (\ref{eq:linearGaussian}) with three stochastic components (Stirrup et al, 2015; Asar et al, 2016),
\nocite{stirrup2015}  \nocite{asar2016}
and none that allow the
three scaling factors to be de-coupled.
  	
\section{A flexible class of non-Gaussian models} 
\label{sec:model}

Our aim in this section is to set out a version of the mixed effects model
\begin{equation}
Y_{ij} = \mv{x}_{ij}^\top \mv{\beta} + \mv{d}_{ij}^\top \mv{U}_i + W_i(t_{ij}) + \sigma Z_{ij}; \quad j=1,\ldots,n_i, \quad i=1, \ldots,m,
\label{eq:three-component}
\end{equation}
that allows Gaussian or non-Gaussian distributional
 specifications of the three stochastic components $\mv{U}_i$, $W_i(t)$ and $Z_{ij}$ to be de-coupled.

 Writing  $\mv{X}$ and $\mv{Z}$ to 
 denote generic vector-valued random variables, the specification will be based on 
 replacing each of the Gaussian distributions with a Normal variance-mean mixture of the form 
 \begin{equation}
 \mv{X} = \mv{\delta} + \mv{\mu} V + \sqrt{V} \mv{\Sigma}^{1/2} \mv{Z},
 \label{eq:ghd}
 \end{equation}
 where $\mv{\delta}$ and $\mv{\mu}$ are vectors of parameters, 
 $\mv{Z} \sim {\rm N}(\mv{0},\mv{I})$, and $V$ is a random variable that takes values on $\R^+$. 
 We need to impose some restrictions on the distribution of $V$ 
 for the inferential algorithms
 that we develop in Section \ref{sec:inference} to be
 practicable.
  For the subject-specific {\it random-effect}, 
 $\mv{U}_i$, and the measurement-specific {\it noise}, $Z_{ij}$, the only necessary restriction is 
 that $V$ has a known distribution. However, to simplify parameter estimation, we shall impose
 the additional restriction that $V|\mv{Z}$ also has a known distribution. For the subject-specific continuous-time 
 {\it stochastic process}, $W_i(t)$, we  use a numerical discretisation of the differential 
 operator (\ref{eq:GaussianSDE}) to generate realisations of  the process. For this reason,
 we need
the distribution to be closed under arbitrary discretisation, which we ensure by 
 requiring the distribution of $V$ to be closed under convolution. Our specific proposals for 
 $\mv{U}_i$, $W_i(t)$ and $Z_{ij}$ are described in more detail below.
 
 \subsection{Noise}

 A flexible choice for the noise term in (\ref{eq:three-component}) is the multivariate generalised 
 hyperbolic (GH) distribution (Barndorff-Nielsen, 1977; Vilca, Balakrishnan and Zeller, 2014). 
 This distribution can be generated from the mixture representation \eqref{eq:ghd} by specifying a
 generalised inverse Gaussian distribution (GIG) for $V$. The density function of the GIG distribution is
 \begin{equation}
 f(x; p, a, b) = \frac{(a/b)^{p/2}}{2K_p\left(\sqrt{ab}\right)} x^{p-1} \exp\left(- \frac{a}{2} x - \frac{b}{2} x^{-1}\right),
 \label{eq:GIG}
 \end{equation}
 where $K_p$ is the modified Bessel function of the third kind, of order $p$, and $a$ and $b$ are positive-valued parameters. 
 We denote this distribution by ${\rm GIG}(p,a,b)$ and refer 
the reader to \cite{jorgensen2012} for more details. An important property of this 
 distribution is that for any $c > 0$, $c V \sim {\rm GIG}(p,a/c,cb)$. 
 Another property that is useful for the construction of the sampling-based
 inferential algorithms that we introduce in 
 Sections \ref{sec:stochastic_gradient} and \ref{sec:subsampling} 
 is that the conditional distribution of $V$ given the observed data is also GIG.

 The GH distribution includes several widely-used distributions as special cases, for example the Student's $t$, generalized asymmetric Laplace (GAL), 
Normal-inverse Gaussian (NIG) and  Cauchy (CH) distributions.
 Specific parameter configurations for the distributions
 of $V$ 
 that give each of these special cases are presented 
 in Table \ref{tab:GIG}.

 \begin{table}[tb]
 \centering
 \fbox{
 \begin{tabular}{ l l  l }
 Distribution of $\mv{X}$  & mixing distribution of $V$           & GIG form of the mixing distribution \\ \hline
 t$(\mv{\delta}, \mv{\mu},\nu)$ 		 & IGam$(\frac{\nu}{2}, \frac{\nu}{2})$ & GIG$(\frac{\nu}{2}, \nu,0 )$\\
 NIG$(\mv{\delta}, \mv{\mu}, a, b)$  & IG$(a,b) $ 							            & GIG$(-\frac{1}{2}, a, b)$ \\
 GAL$(\mv{\delta}, \mv{\mu}, p, a)$  & Gam$(p,a) $						              & GIG$(p, 2a, 0)$\\
 CH$(\mv{\delta}, \mv{\mu}, b)$ 	   & IGam$(\frac{1}{2}, \frac{b}{2})$ 		& GIG$(-\frac1{2}, 0, b)$
 \end{tabular}} 
 \caption{Some special cases  of the GH, their mixing distributions and their corresponding GIG forms. 
 Gam indicates the Gamma family of distributions, IGamma the inverse Gamma,  IG the inverse Gaussian.}
 \label{tab:GIG}
 \end{table}

 Since the measurement noise is univariate, we can write the mixture representation \eqref{eq:ghd} as 
 \begin{equation}
 \label{eq:univariateGH}
 Z_{ij} = \delta + \mu V_{ij} + \sqrt{V_{ij}} \sigma Z_{ij}^*,
 \end{equation}
 where $Z_{ij}^* \sim {\rm N}(0,1)$. To maintain the interpretation of $\sigma^2$ as the variance of the noise,
 we constrain the values of the GIG parameters $a$, $b$ and $p$, so that ${\rm E}[V_{ij}]=1$ 
 if ${\rm E}[V_{ij}]$ exists, and the mode of the distribution is one otherwise. 
 We further set $\delta = \mu = 0$ to ensure that
the measurement noise is symmetric with ${\rm E}[Z_{ij}] = 0$.

 An alternative to \eqref{eq:univariateGH} is to attach a single random variable $V_i$ to all of the noise terms
$Z_{ij}$ on the $i$th subject, 
 i.e. $Z_{ij} = \delta + \mu V_{i} + \sqrt{V_{i}}\sigma Z_{ij}^*$. 
 The distribution of $V_i$ can then be interpreted
as a random effect distribution for patient-specific measurement noise variance. Note, in particular,
 that this  introduces stochastic dependence
 between $Z_{ij}$ and $Z_{ij^{\prime}}$ 
for $j \neq j^{\prime}$.

 \subsection{Random effects}

 For the random effects, 
we let $\mv{U}_i = \mv{\delta} + \mv{\mu} V_{i} + \sqrt{V_{i}}\Sigma^{1/2} \mv{U}_{i}^*$, 
 where $V_i$ is a unit-mean GIG random variable and $\mv{U}_{i}^* \sim {\rm MVN}(\mv{0}, \mv{I})$ with
 $\mv{I}$ the identity matrix. 
 To allow skewness in the distribution
of the random effects we do not require that $\mv{\mu}=\mv{0}$, but we
then  ensure that ${\rm E}[\mv{U}_i] = \mv{0}$ by 
 setting $\mv{\delta} = -\mv{\mu}$.
 
 \subsection{Stochastic process}

 The simplest  way to introduce  a non-Gaussian  stochastic process term in (\ref{eq:three-component})
 would again be to include a subject-specific scaling, i.e. $W_i(t) = V_i W_i^*(t)$, where 
 $V_i$ again follows a unit-mean GIG distribution. However, 
 this approach would not be able to capture interesting within-subject departures from Gaussian behaviour, 
 e.g. jumps or asymmetries in the sample paths of $W_i(t)$. To provide the required flexibility, we 
 instead use
 non-Gaussian generalisations of the stochastic differential equation \eqref{eq:GaussianSDE}.
 Specifically, we propose modelling the $W_i(t)$ as independent copies of the solution to 
 \begin{equation}\label{eq:nonGaussSDE}
 \mathcal{D} W_i(t) = dL_i(t),
 \end{equation}
 where the $L_i$ are independent copies
 of a non-Gaussian L\'{e}vy process, i.e. a process with independent and stationary
 increments. In practice, we work with a discretised version of (\ref{eq:nonGaussSDE}), 
 for which \cite{bolin2014} showed that a type-G L\'evy process for $L_i(t)$ is a suitable candidate. 
 The implication is that the increments of $L_i$ have a distribution that corresponds to the specification 
 given by \eqref{eq:ghd}.
  
 One approach would therefore be to choose the distribution of $V$ in \eqref{eq:ghd} as a GIG distribution, 
 which would yield the Hyperbolic processes of \cite{eberlein2001application}. However, as noted earlier, we require the distribution of $V$ to be 
 closed under convolution \citep{wallin2015}. Also, the stochastic gradient methods 
 for parameter estimation to be 
introduced in Section \ref{sec:inference} require
 sampling from the conditional distribution of $V$ given all other components in the model. 
 Within the GH family, the NIG, GAL, and Cauchy distributions are the only ones that meet 
 these requirements \citep{podgorski2016convolution}. Using any of these distributions for the increments 
 of $L_i$ in \eqref{eq:nonGaussSDE} results in models with the same covariance structure 
 as if $L_i$ were Gaussian, but with more general marginal distributions. The NIG choice 
 makes $L_i$ a NIG process \citep{barndorffnielsen1997a}, which has been used in financial modeling; see 
\citet{barndorffnielsen1997b}, \citet{Bibby2003}, 
\citet{tankov2003financial} and
\citet{eberlein2001application}. We will focus on 
 the NIG case in order to keep the presentation brief, but the modifications needed to cover the
 GAL and Cauchy counterparts are straightforward.

 For computational purposes, we use a discretised version of the stochastic differential 
 equation (\ref{eq:nonGaussSDE}) as follows;
see also \citet{lindgren2008}. 
Firstly, denote by $\scal{f}{g}$ the standard 
 inner product on $\R$,
 $\scal{f}{g}=\int f(t)g(t)dt$. We restrict $W(t)$ to a finite interval, 
 $0 \leq t \leq t_0$, and impose Neumann boundary conditions, $dW(0) = dW(t_0) = 0$. 
 The so-called {\it weak form} of (\ref{eq:nonGaussSDE}) is a function of $W(t)$ that 
 satisfies the equation,
 \begin{equation}
 \scal{\psi}{\mathcal{D}W} = \scal{\psi}{dL},
 \label{eq:weak}
 \end{equation}
 for a specified set of {\it test functions} $\psi(t)$.
 We now use the following {\it low-rank} approximation,
 \begin{equation}\label{eq:basisexp}
 W(t) = \sum_{k=1}^{K} \phi_k(t)  W_k,
 \end{equation}
 where $\mv{W} = (W_1, \ldots, W_{K})$ is a vector of random variables 
and the $\phi_k(t)$ are basis functions. We use a set of piecewise linear basis functions such that 
 $$
 \phi_1(t) = \begin{cases}
 1-\frac{t-s_1}{s_{2}-s_1}, & s_{1} < t < s_{2}, \\
 0, & \mbox{otherwise},
 \end{cases}
 $$
 $$ 
 \phi_K(t) = \begin{cases}
 \frac{t-s_{k-1}}{s_k-s_{k-1}}, & s_{k-1} < t < s_k, \\
 0, & \mbox{otherwise},
 \end{cases}
 $$
 and, for $k=2, 3, ... , (K-1)$,
 $$
 \phi_k(t) = \begin{cases}
 \frac{t-s_{k-1}}{s_k-s_{k-1}}, & s_{k-1} < t < s_k, \\
 1-\frac{t-s_k}{s_{k+1}-s_k}, & s_{k} < t < s_{k+1}, \\
 0, & \mbox{otherwise},
 \end{cases}
 $$
 where $0 = s_1 < s_2 < \ldots < s_{K-1} < s_{K} = t_0$.  

 We use the Galerkin finite element method to compute the stochastic weights, $W_k$. 
 This consists of setting all the test functions 
 to the basis functions, i.e. $\psi_k = \phi_k$ for all $k$, 
 and computing the $W_k$ by solving 
 the system of equations defined by equation \eqref{eq:weak}, i.e. 
 $$
 \mv{K}\mv{W} = \mv{L},
 $$
 where $L_k = \scal{\phi_k}{dL}$, and $K_{kk^\prime} = \scal{\phi_k}{\mathcal{D}\phi_{k^\prime}}$ is a discretised version
 of the differential operator $\mathcal{D}$. For the NIG version of the model, we approximate the distribution of $L_k$ by 
 $$
 L_k = h_k \delta  + \mu V_k + \sqrt{V_k} Z_k,
 $$
 where $Z_k \sim {\rm N}(0, 1)$, $h_k= \scal{\phi_k}{1}$, and $V_k \sim {\rm IG}(\nu, h_k^2 \nu)$ \citep{bolin2014}. 
 It follows that the distribution for the stochastic weight-vector $\mv{W}$ conditional on $V$ can be written as 
 $$
 \mv{W}|\mv{V} \sim {\rm N}(\mv{K}^{-1} \left( \mv{h}^{\top} \delta + \mv{V}^{\top} \mu\right) , \mv{K}^{-1}\diag(\mv{V})(\mv{K}^{-1})^{\top}).
 $$
 Since the parameter $\nu$ determines the value of
 $b$ in the NIG distribution, it controls the tails of the marginal distribution of the process. The limit when $\nu \rightarrow 0$ is the Cauchy process, whereas the limiting case $\nu \rightarrow \infty$ is a Gaussian process. These are exactly the properties we need in order to use our
 likelihood-based methods to  assess whether a standard, and undeniably convenient, 
 Gaussian assumption for any or all of the stochastic components of (\ref{eq:three-component}) is adequate.  

\subsection{Similarity between densities}

 The full GH family of distributions is
difficult to fit to data, 
 because the full log likelihood surface is largely flat, which 
 makes the model parameters almost non-identifiable.
 The problem persists for some of the sub-families discussed above.
 For example, a NIG distribution converges to a Cauchy distribution as $a\rightarrow 0$ and
 to a Gaussian distribution as $a\rightarrow \infty$ and  $b \rightarrow \infty$ at the same rate.
 Recognising these limiting cases is important in practice, 
 since typically the densities are numerically unstable at the edges of 
 the parameter space. 
 Moreover,  handling the NIG distribution is computationally more demanding than the Gaussian distribution.
 Rules for switching between the distributions, or 
 equivalently setting finite boundaries to the parameter space, require some guidance.
 This can be achieved by the total variation (TV) distance between pairs of densities. 
 For illustration, we consider the TV distance between two symmetric, 
 zero-mean distributions. For instance, to compare the NIG distribution 
 for fixed $a$ with
 the Cauchy distribution, we calculate
 $$TV_{NIG,CH}(b_{CH}, a,b_{NIG} ) = \min_{b_{CH}} \int |{\rm CH}(x; 0,0,b_{CH}) - {\rm NIG}(x;0,0,a,b_{NIG}) | dx.$$
 To simplify the calculations needed to find the Cauchy 
 distribution ${\rm CH}(0,0,b_{CH})$ that 
is closest to the ${\rm NIG}(0,0,a,b_{NIG})$ distribution 
 we use the proposition below,
 which shows that it suffices first
to find the Cauchy distribution
 closest to ${\rm NIG}(0,0,a,1)$,
then rescale the shape parameter by $b_{NIG}$.
 \begin{prop*}
 Let $f_{s}(x)$ and $g_{h}(X)$ be two distributions with respect to the Lebesgue measure, with 
 scaling parameters $s$ and $h$. Then,
 $TV(f_s,g_h) = TV(f_{s/c},g_{h/c} )$ for $c>0$.
 \end{prop*}
 \begin{proof}
 First note that
 $$
 TV(f_{s/c},g_{h/c} ) = \frac{1}{2} \int |f_{s/c}(x) - g_{h/c}(x)| dx = \frac{c}{2} \int |f_{s}(cx) - g_{h}(cx)| dx.
 $$
 Now use integration by substitution with respect to
$\phi(x) = \frac{x}{c}$ to give
 $$
 \frac{c}{2} \int |f_{s}(cx) - g_{h}(cx)| dx = \frac{1}{2} \int |f_{s}(x) - g_{h}(x)| dx = TV(f_s,g_h).
 $$
 \end{proof} 
 
The proposition can also be used to compare NIG and Gaussian distributions. 
Figure \ref{fig:TV} shows the TV distances between the NIG and 
 Cauchy, and between the NIG and Gaussian, as functions of $a$. 
 For $a = 0.001$, the TV distance between the NIG and Cauchy is less than that between 
 two Bernoulli distributions whose probabilities differ by $0.002$. The same
 applies to the TV distance between the NIG and Normal when $a=250$. 
 This suggests that setting the boundary of $a$ at these values, i.e., switching from the NIG to the
 Gaussian when $a>250$ and to the Cauchy when $a<0.001$, is a conservative strategy for
parameter estimation, or for prediction within or close to the observed range of the  data. 
Since the differences between the distributions are in their tails, prediction of extreme events could be affected by the switching even when one cannot tell empirically which density should be used. For example, the NIG distribution has exponential  tails for all values $a$, whereas the Cauchy has polynomial tails.
 \begin{figure}[t]
 \centering
 \includegraphics[width=0.4\linewidth]{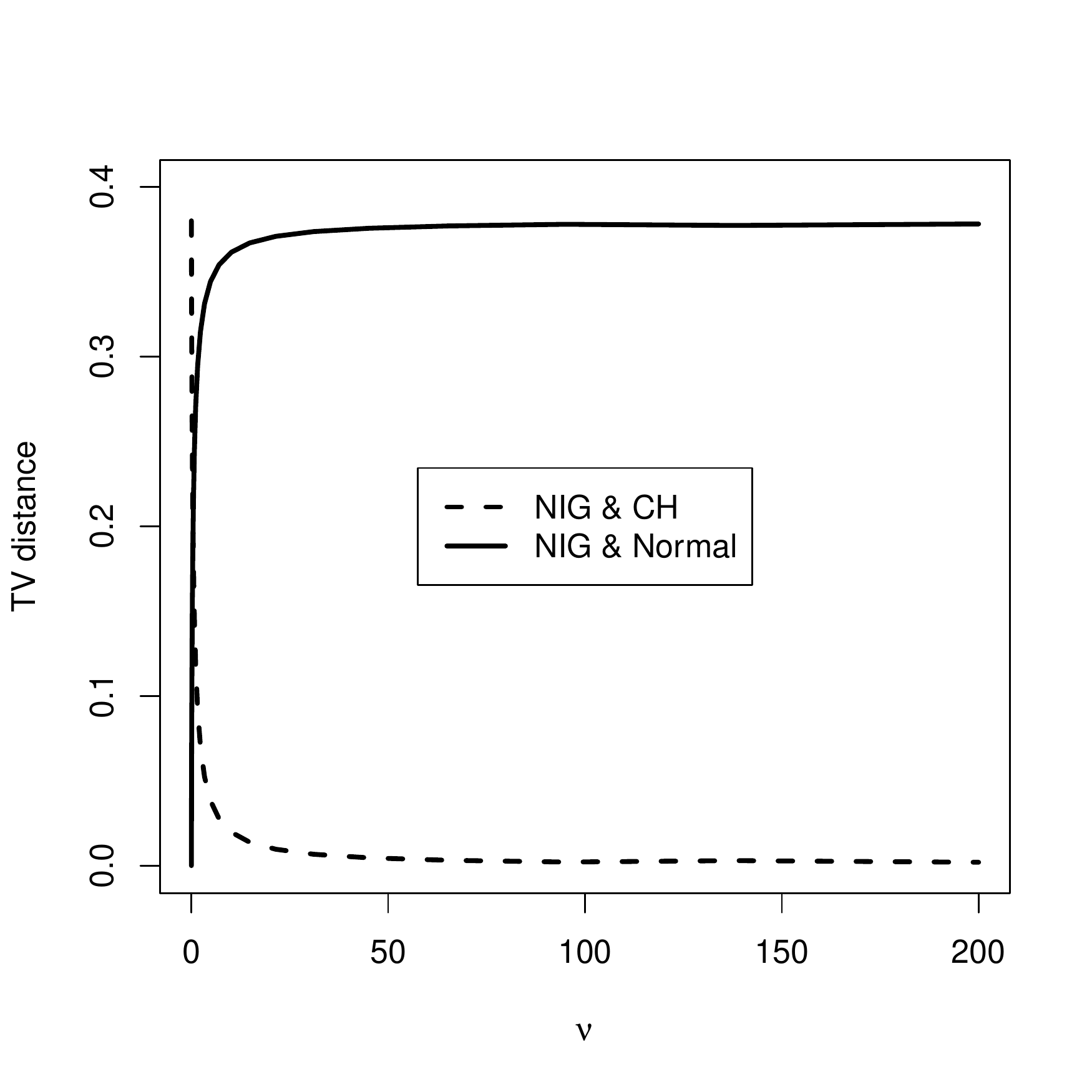}
 \caption{The dahsed line is the TV distance between NIG and Cauchy for varying $\nu$,  
          the solid between NIG and Normal.}
 \label{fig:TV}
 \end{figure}

 \section{Likelihood-based inference} 
 \label{sec:inference}

 \subsection{Hierarchical representation}
 \label{sec:HierRep}

 Our specification of a Normal-variance mixture  for
 each of the stochastic components of (\ref{eq:three-component}) 
 makes likelihood-based inference practicable via the following hierarchical
 representation of the model. For subject $i$, let $\mv{V}_i^Z$, $V_i^U$, and $\mv{V}_i^W$ 
 denote the stochastic variance factors corresponding to the noise, random effect and stochastic process 
 components of (\ref{eq:three-component}), and write  $\mv{Y}_i = (Y_{i1}, \ldots, Y_{in_i})^{\top}$ for the
 corresponding set of repeated measurements. Let 
 $W_i =\{ W_{ik}: k = 1, \ldots, K\}$ 
 be the stochastic weight vector for the $i$th subject in
 the approximation of $W_i(t)$ given by \eqref{eq:basisexp}, 
 and $\mv{A}_i$ the $n_i$ by $K$ matrix 
 with $(j,k)$th element $\phi_k(t_{ij})$. Write 
 $\mv{x}_i$  and $\mv{d}_i$  for  the matrices with $j$th row
 $\mv{x}_{ij}^\top$ and $\mv{d}_{ij}^\top$, respectively. Finally, let
 $\mv{\Theta}$ denote the complete set of model parameters. 
 The model for the $i$th subject then has the following hierarchical representation
 \begin{align*}
 \mv{Y}_i|\mv{W}_i, \mv{U}_i,\mv{V}_i^Z, \mv{\Theta} &\sim {\rm N}(\mv{x}_i\mv{\beta} + 
 \mv{d}_i \mv{U}_i + \mv{A}_i\mv{W}_i, \sigma^2 \diag(\mv{V}_i^Z)),\\
 \mv{U}_i | V_i^U,\mv{\Theta} &\sim {\rm N}(-\mv{\mu}^U + \mv{\mu}^U\cdot V_i^U , V_i^U \mv{\Sigma})  \\ 
 \mv{W}_i|\mv{V}_i^W,\mv{\Theta} &\sim {\rm N}(\mv{K}^{-1} \left(- \mv{h} \cdot \mu^W + \mu^W\cdot \mv{V}_i^W \right),\mv{K}^{-1}\diag(\mv{V}_i^W)(\mv{K}^{-1})^{\top}).
 \end{align*}
 These collectively determine the contribution of the $i$th subject to the log-likelihood, $L(\mv{\Theta}; \mv{y}_{i})$. 
 As the vectors $\mv{Y}_i$ from the $m$ subjects are independent, the overall log-likelihood is
 $$
 L(\mv{\Theta}; \mv{y}) = \sum_{i=1}^m L(\mv{\Theta}; \mv{y}_{i}).
 $$

 \subsection{Stochastic gradient estimation}
 \label{sec:stochastic_gradient}
 
 The computations required for  maximum likelihood estimation are 
 cumbersome for problems that involve
 longitudinal data-sets with 
 large numbers of subjects and repeats, even using
  the computationally efficient approximation (\ref{eq:basisexp}). 
 Our proposed algorithm for ML estimation therefore uses 
 a stochastic gradient method that calculates the gradient 
 of the objective function at each step of the maximisation 
 by sub-sampling.

 A stochastic gradient method for the general problem of 
 minimising an objective function $f(\mv{\Theta})$
 starts with an initial guess $\mv{\Theta}^{(0)}$, and then 
 iteratively updates $\mv{\Theta}$ according to
 \begin{equation}
 \mv{\Theta}^{(n + 1)} = \mv{\Theta}^{(n)} + \alpha_n Q_n(\mv{\Theta}^{(n)}),
 \label{eq:sg}
 \end{equation}
 where $Q_n(\mv{\Theta})$ is a random variable such that 
 ${\rm E}[Q_n(\mv{\Theta})] = \nabla f(\mv{\Theta})$ and $\alpha_i$ 
 is a sequence of positive numbers such that 
 $\sum_{n = 1}^{\infty} \alpha_n = \infty$ and 
 $\sum_{n = 1}^{\infty} \alpha_n^2 < \infty$. 
 Under mild regularity conditions, the resulting
 sequence $\mv{\Theta}^{(n)}$ converges to a stationary point
 of $f(\mv{\Theta})$ \citep{Kushener2003, andrieu2007stability}.
 
 For maximum likelihood estimation, $f(\mv{\Theta}) = -L(\mv{\Theta}; \mv{y})$.
 If the data-set contains a large number of subjects we
 use only 
 a small, randomly sampled subset  in each iteration to generate an 
 efficient stochastic gradient method. For this purpose, 
 $\nabla L(\mv{\Theta}; \mv{y})$ can be replaced by the random variable
 \begin{align}
 Q_n(\mv{\Theta}) = \nabla L_s(\mv{\Theta}; \mv{y}) =  s \sum_{i = 1}^m \nabla L(\mv{\Theta}; \mv{y}_{i}) J_i,
 \label{eq:subsample}
 \end{align}
 where the $J_i$ are independent Bernoulli random variables with ${\rm P}(J_i=1) = 1/s$. 
 Since ${\rm E}[\nabla L_s(\mv{\Theta}; \mv{y})] = \nabla L(\mv{\Theta}; \mv{y})$ for any $s$,
 the resulting stochastic gradient method (\ref{eq:sg}) will converge to a stationary point of the log-likelihood. 
  Our experience, for example with the two case-studies that we describe in Section \ref{sec:examples},
 has been that for data-sets containing a large number of subjects, often we need TO access only a small proportion of the available 
 measurement sequences $\mv{Y}_i$
  at each iteration in order to estimate the parameters reliably. Our proposed algorithm
 therefore becomes
 $$
 \mv{\Theta}^{(n+1)} = \mv{\Theta}^{(n)} + (\alpha_n/m) \nabla_{\mv{\Theta}}L_s(\mv{\Theta}; \mv{y}),
 $$
 where the multiplier $\alpha_n/m$ highlights that for numerical stability the step length 
 should scale with the number of subjects in the complete data-set. 

 For our non-Gaussian models, an additional complication
 is that the likelihood is not available in an explicit form. 
 However, using Fisher's identity \citep{dempster1977} we can 
 compute the gradient of the log-likelihood without 
 computing the log-likelihood itself. 
 For all versions of our model, the log-likelihood conditional
 on the variance components, 
 $\mv{V} =  \{\mv{V}_{i}^W, \mv{V}_i^U, \mv{V}_i^Z\}_{i=1}^m$,
 is Gaussian and thus explicit. Fisher's identity then gives
 $$
 \nabla L_s(\mv{\Theta}; \mv{y}) = {\rm E}_{\mv{V}}(\nabla L_s(\mv{\Theta}; \mv{y},\mv{V})|\mv{y},\mv{\Theta}),
 $$
 where $L_s(\mv{\Theta}; \mv{y})$ is the augmented likelihood and is explicitly available, 
 since $\mv{Y}, \mv{V}| \mv{\Theta}$ is Gaussian and $V|\mv{\Theta}$ is GIG. 

 The expectation is not, in general, explicit but can be 
 approximated by Monte Carlo sampling from
 the conditional distribution $\mv{V}|\mv{y};\mv{\Theta}$. 
 We use a Gibbs sampler and iterate between sampling from the
conditional distributions
 $\mv{V}|\mv{X},\mv{Y}; \mv{\Theta}$ and
 $\mv{X}|\mv{V}, \mv{Y}; \mv{\Theta}$, where
 $\mv{X}$ denotes 
 all the conditional Gaussian components, i.e. 
 $\mv{X} = \{\mv{U}_{i}, \mv{W}_i\}_{i=1}^m$. 
 Convergence of algorithms of this kind is studied in \cite{andrieu2007stability}.

 When using stochastic gradient optimization to maximise over many parameters, 
 it is important to scale the gradient by a pre-conditioner to give 
 a Newton-like iteration:
 \begin{align}
 \label{eq:grads}
 \mv{\Theta}^{(n + 1)} = \mv{\Theta}^{(n)} + \alpha_n \mv{I}^{-1} Q_n(\mv{\Theta}^{(n)}).
 \end{align}
 One option for the pre-conditioner is 
 \begin{equation}
 \mv{I}^*(\mv{\Theta}) = -{\rm E}_{\mv{V}}(\nabla \nabla L_s(\mv{\Theta}; \mv{y},\mv{V})|\mv{y},\mv{\Theta}).
 \label{eq:conditioner}
 \end{equation}
 Calculation of $\mv{I}^* (\mv{\Theta})$ is typically easy,  
 since $\nabla \nabla L_s(\mv{\Theta}; \mv{y},\mv{y})$ is often explicit 
 and can be calculated at the same time as the gradient. 
  \cite{lange1995} describe the connection between (\ref{eq:conditioner}) and the EM algorithm. 
 However, if the same variables 
 are used for the Monte Carlo estimates of the expectations in
 $\mv{I}_{.}(\mv{\Theta})$ and $Q_i(\mv{\Theta})$, the joint updating 
 step \eqref{eq:grads} will be biased due to correlation between 
 the two estimated quantities. 
 A  pre-conditioner
 that is less  problematic numerically, and is
 often explicitly available,
 is the {\it complete} Fisher Information (cFIM),
 \begin{equation}
 \mv{I}_{cFIM}(\mv{\Theta}) = -{\rm E}_{\mv{V},\mv{Y}}(\nabla\nabla L_s(\mv{\Theta}; \mv{y},\mv{V})|\mv{\Theta}).
 \label{eq:complete_Fisher}
 \end{equation}
 Note that in (\ref{eq:complete_Fisher}),
 the expectation is taken over both $\mv{Y}$ and $\mv{V}$. 
 The {\it standard} Fisher information matrix,
 $$
 \mv{I}_{FIM}(\mv{\Theta}) = -{\rm E}_{\mv{Y}}(\nabla\nabla L_s(\mv{\Theta}; \mv{Y})|\mv{\Theta}),
 $$
 is seldom explicit and thus cannot be used as a pre-conditioner. 
 However, we do need to estimate either the standard or 
 the {\it observed} Fisher information matrix,
 $$
 \mv{I}_{oFIM}(\mv{\Theta}) = -\nabla \nabla L_s(\mv{\Theta}; \mv{Y}),
 $$
 in order to calculate confidence intervals
 for the estimated parameters. 
 We estimate $\mv{I}_{oFIM}(\mv{\Theta})$ using Louis's identity \citep{Louis1982},
 \begin{equation}
 \mv{I}_{oFIM}(\mv{\Theta}) = - {\rm E}_{\mv{V}}(\nabla \nabla L_s(\mv{\Theta}; \mv{y},\mv{V})|\mv{y},\mv{\Theta}) - \pV_{\mv{V}}\left[\nabla L_s(\mv{\Theta}; \mv{y},\mv{V}) \nabla L_s(\mv{\Theta}; \mv{y},\mv{V})^T|\mv{y},\mv{\Theta}    \right]. 
 \label{eq:louis}
 \end{equation}
 Both terms on the right-hand side of
(\ref{eq:louis}) can be estimated by Monte Carlo sampling, 
 as proposed for $\nabla L_s(\mv{\Theta}; \mv{y})$ in
 \eqref{eq:subsample}.

 We could estimate $\mv{I}_{FIM}(\mv{\Theta})$ 
 by an additional sampling step, using the fact that 
 $\mv{I}_{FIM}(\mv{\Theta}) = {\rm E}_{\mv{Y}}[\mv{I}_{oFIM}(\mv{\Theta})|\mv{\Theta}]$. 

 \subsection{Sub-sampling with fixed effects: the grouped sub-sampler}
 \label{sec:subsampling}


 An issue with the sub-sampling method described 
 in \eqref{sec:stochastic_gradient} is that 
 the sub-sampled matrices of covariates, $\mv{x}_i$,
 may not be of full rank. If this is the case, none of the pre-conditioners above can be used.
  On the other hand, regular sub-sampling without any pre-conditioners 
 may result in large Monte Carlo variation in the estimated gradient. 
 The cystic fibrosis case-study that we shall describe in Section \ref{sec:CF}
 provides an example, where the data are stratified into birth cohorts whose effects are
important, but one of the cohorts contains only
 seven patients. This issue is related to sub-sampling for S-estimation algorithms in linear regression models
\citep{koller16}. Nonetheless, we could not find a satisfactory solution 
 in the literature that could be applied in the current context. 
 To address the issue, we therefore introduce the following sub-sampling procedure, 
 which we  call the {\it grouped sub-sampler}. The procedure 
 first builds $k+1$ groups of subjects, 
 $\mathcal{G}_0, \mathcal{G}_1, \ldots, \mathcal{G}_k$, 
 in such
 a way that the matrices $\sum_{i \in \mathcal{G}_k} \mv{x}_i$ 
 have full column ranks for $k\geq 1$. The procedure for forming the 
 groups is described in the pseudo-code in Algorithm \ref{alg:groups}.  
 Let $m_g$ be the number of subjects in group $g$, 
 and write $\bar{m} = k^{-1} \sum_{g=1}^k m_g$.
 A sub-sampling step  selects approximately  $M$ subjects by first 
 selecting all subjects in $r < k$  groups chosen at random from 
 $\mathcal{G}_1, \ldots, \mathcal{G}_k$, then 
 adding $M - \bar{m} \times m_g$ subjects chosen at random 
 from $\mathcal{G}_0$. The expected number of sub-sampled 
 subjects is then $M$, and the  matrix of covariates for the sub-sampled 
 subjects has full column rank. To obtain an unbiased estimate of the 
 gradient, we then assign weights $k/m_g$ to the subjects sampled 
 from the groups $\mathcal{G}_1, \ldots, \mathcal{G}_k$, and 
 weight $m_0/(M - \bar{m} \times m_g)$ to those sampled 
 from $\mathcal{G}_0$. 

 \begin{algorithm}[t]
	\caption{Group formation for the grouped sub-sampler. }
	\label{alg:groups}
	\begin{algorithmic}[1]
		\Procedure{group-formation}{$\mv{x}_1,\ldots, \mv{x}_M$}
		\State $\mathcal{I} \gets \{1,\ldots, M\}$
		\State $k \gets 1$
		\While{$|\mathcal{I}|>0$}
		\State $\mathcal{G}_k \gets$ \Call{create-group}{$\mathcal{I}$,$\mv{x}_1,\ldots, \mv{x}_M$}
		\If{$rank(\sum_{i\in \mathcal{G}_{k}}\mv{x}_{i} \mv{x}^T_{i}) = columns(\mv{x}_1)$} \Comment{$columns(\mv{x}_1)=\#$  covariates}
			\State $\mathcal{I} \gets \mathcal{I}\setminus \mathcal{G}_k$
			\State $k \gets k+1$		\Else
		\State $\mathcal{G}_0 \gets \mathcal{I}$
		\State $\mathcal{I} \gets \emptyset$
		\EndIf
		\EndWhile
		\State \textbf{return} $\mathcal{G}_0,\ldots, \mathcal{G}_k$
		\EndProcedure
		\Procedure{create-group}{$\mathcal{I}$,$\mv{x}_1,\ldots, \mv{x}_M$}
		\State $\mathcal{G} \gets \mathcal{I}_1$ 
		\State $\mathcal{I} \gets \mathcal{I}\setminus \mathcal{I}_1$
		\While{$rank(\sum_{i\in \mathcal{G}}\mv{x}_{i} \mv{x}^T_{i}) < columns(\mv{x}_1)$ and  $|\mathcal{I}|>0$}
		\If{$rank(\mv{x}_{\mathcal{I}_1} \mv{x}^T_{\mathcal{I}_1}+\sum_{i\in \mathcal{G} }\mv{x}_{i} \mv{x}^T_{i})>rank(\sum_{i\in \mathcal{G}}\mv{x}_{i} \mv{x}^T_{i})$}
		\State $\mathcal{G} \gets \mathcal{G}\cup\mathcal{I}_1$
		\EndIf
		\State $\mathcal{I} \gets \mathcal{I}\setminus \mathcal{I}_1$
		\EndWhile
		\State \textbf{return} $\mathcal{G}$
		\EndProcedure
	\end{algorithmic}
\end{algorithm}

 \section{Case-studies} 
 \label{sec:examples}

 \subsection{Natural progression of lung function in cystic fibrosis patients} \label{sec:CF}
 
 Our first application uses data on the lung function of cystic fibrosis  
 patients, taken from the Danish Cystic Fibrosis register. The patients 
 are all aged over 5 years, and entered the database between 1969 and 2010. 
 The outcome variable is \%FEV1 (per cent predicted forced expiratory
 volume in 1 second), a measure of lung function that is widely used as a 
 descriptor of disease severity \citep{davies2009}. The data, previously 
 analysed by \cite{taylor_robinson2012}, contain 70,448 measurements of \%FEV1
 on 479 patients with follow-up times approximately one month apart. For the 
 analysis reported here, three patients who provided only one \%FEV1 measurement 
 have been excluded. Hence, 476 patients are available for the current analysis. 
 Available covariates are: sex, age, birth cohort (decadal), 
 presence/absence of pancreatic sufficiency, 
 presence/absence of diabetes mellitus, 
 and years after pseudomonas infection. 
 The number of repeated measures per patient ranges between 2 and 597 with a median of 101.5. 
 Total follow-up times ranged between 0.1 and 31.5 years with a median of 10.5. 
 Of the 476 patients, 233 (48.9\%) are female, 20 (4.2\%) have pancreatic sufficiency, 
 14 (2.9\%) have diabetes. Baseline ages range between 5.0 and 48.1 years with a median of 7.0. 
 Cohort numbers are 7 (1.5\%), 42 (8.8\%), 109 (22.9\%),  105 (22.1\%), 141 (29.6\%) and 
 72 (15.1\%) for birth cohorts of 1948--1957, 1958--1967, $\ldots$, 1998--2007, respectively.   
 Baseline \%FEV1 values range between 10.4 and 140.3 with a mean of 78.5.  
 Figure \ref{fig:CF} shows  traces for six patients, chosen to illustrate
 a range of total follow-up times and  patterns of the outcome variable,
 \%FEV1. 
	
 \begin{figure}[t]
 \centering
 \includegraphics[width = 0.6\linewidth]{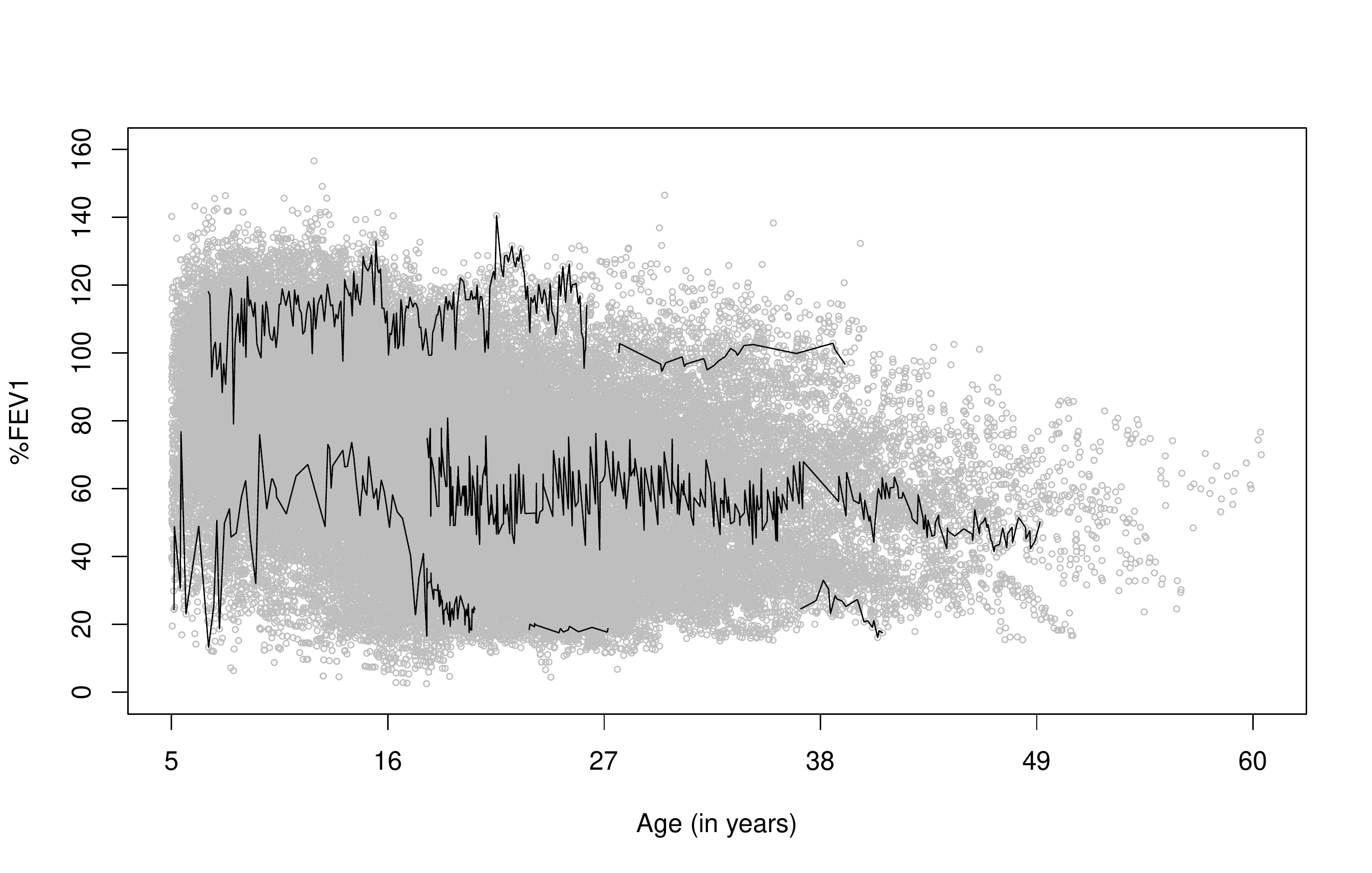}
 \caption{\%FEV1 measurements against age (in years) in background as a grey scatter-plot. 
 Data on six patients are highlighted as black lines.}
 \label{fig:CF}
 \end{figure}
  
 Fitting a model to these data serves two purposes. The first is to characterise the mean
 response profile of  lung function in cystic fibrosis patients, adjusted for relevant covariates. The 
 second is to quantify the extent to which a subject's early results are predictive
 of their long-term prognosis. 

 We let $Y = $ \%FEV1 and specify mixed effects models that fall within the 
 general framework of (\ref{eq:three-component}). Specifically we consider   
 \begin{equation}
 \label{eq:thorax_model}
 Y_{ij}= \mv{x}_{ij}^{\top} \mv{\beta} + U_i + W_i(t_{ij}) + Z_{ij},
 \end{equation}
 where 
the
 $W_i(t)$'s are mutually independent, continuous-time stochastic processes. We model this process
 as the solution to the stochastic differential equation 
 $\left(\kappa^2 - \frac{d^2}{dt^2}\right)^{1/2} W_i(t) = dL_i(t)$, which implies that
  $W_i(t)$ has an exponential covariance function, 
	as in  \citet{taylor_robinson2012}. 

 In this example cohort effects are substantial, reflecting general improvements in the treatments
 available to CF patients over the time-period concerned. This, coupled with the small 
 numbers of patients in some cohorts (e.g. 7 patients in 1948--1957), 
 explains why the grouped sub-sampler 
 described in Section \ref{sec:subsampling} is needed.

 To illustrate the effect of the sub-sampling, we first fit a Gaussian model, i.e  assuming Gaussian distributions for 
$U_i$, $W_i(t)$ and $Z_{ij}$, with  
 and without sub-sampling. In the former case,
we sub-sample 20\% of the patients, i.e. 96 out of 476. 
 The resulting parameter tracks of the optimiser can be seen in Figure \ref{fig:paths}.  
 In this example, there are $k=7$ sub-sampling groups, with an average group size of 
 $8$ subjects, and two groups were sampled at each iteration. The running time for 
 the $20,000$ iterations scales linearly with $M$, the number of patients 
 sub-sampled at each iteration. In this example,
sub-sampling reduced computing time by a factor of
almost five. The variances of the 
 sub-sampled estimates are relatively higher, 
 but the final parameter estimates are almost identical.

 \begin{figure}[t]
 \centering
 \includegraphics[width=0.9\linewidth]{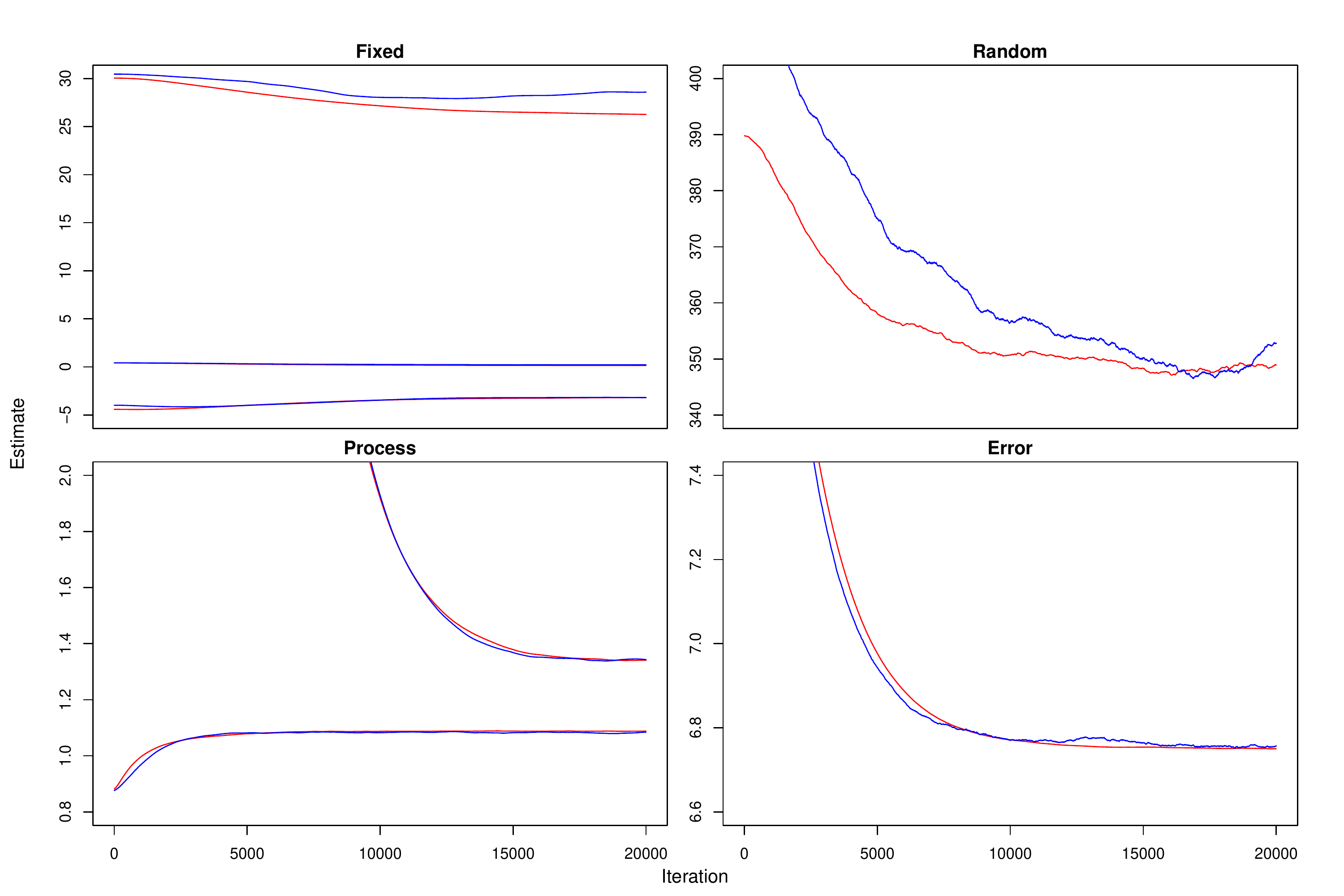}
 \caption{Stochastic gradient estimation paths. 
Red paths show results without sub-sampling, blue paths show results obtained by sub-sampling 20\% of the patients in each iteration.
Upper left: three of the 
 fixed effects parameters; upper right: random effects variance parameter; 
 bottom left: process parameters; bottom right: measurement error variance parameter.
Note that the lower pair of red and blue paths
belong to the parameter $\kappa$ and 
are magnified by a factor of 10 for 
clearer visualisation.}
 \label{fig:paths}
 \end{figure}

 To assess the
suitability of the Gaussian distributional assumption 
we inspected
quantile-quantile (QQ) plots 
 of the standardised marginal residuals, 
  $Y_{ij} - \mv{x}_{ij}^{\top} \mv{\beta}$ .  
 The plot (not shown here) suggests some departure
 from the Gaussian, but 
 as each marginal residual is composed of
 $U_i$, $W_i(\cdot)$ and $Z_{ij}$, the QQ plot is
  not able to detect the source of the departure. 
 We therefore fit the 
 model with a NIG assumption for 
  each of the $U_i$, $W_i(\cdot)$ and $Z_{ij}$ 
 components. The respective estimates of $\nu$ are 83.77, 
 0.34 and 0.48, which indicate  the extent to which each component
 appears to exhibit non-Gaussian behaviour.  
 The estimates of the fixed effect parameters, $\mv{\beta}$, 
 for Normal and NIG models are shown in Table \ref{tab:FEVpars}. 
 Standard error estimates, obtained using standard Fisher matrix, are generally lower under NIG than
 under Gaussian distributional assumptions. 
 With regard to the statistical significance or otherwise of the estimates, 
 we report two p-values, p-lower and p-upper, by taking into account  
 Monte Carlo error in the parameter and standard error estimates. 
 Lower and upper p-values indicate the same judgement on significance, 
 except for pancreatic sufficiency under NIG. Note that 
 this variable is highly unbalanced, 
 with only 20 positives out of 476. 
 Consequently, 
 Normal and NIG models agree on the significance of the estimates, 
 except for the interaction of age and cohort 1958. 
 
\begin{table}[t]
 \centering
 \fbox{
\scalebox{0.9}{
 \begin{tabular}{lrrrrrrrr}
 \multicolumn{1}{c}{} &\multicolumn{4}{c}{Normal} &  \multicolumn{4}{c}{NIG} \\
\cmidrule(r){2-5} \cmidrule(r){6-9}
 & Estimate & SE & p-lower & p-upper & Estimate & SE & p-lower & p-upper \\
\cmidrule(r){2-5} \cmidrule(r){6-9}
 Intercept                         & 68.11 & 0.66 & $<$0.001 & $<$0.001 & 70.51 & 0.71 & $<$0.001  & $<$0.001\\
 Diabetes                          & -3.10 & 0.44 & $<$0.001 & $<$0.001 & -1.82 & 0.38 & $<$0.001  & $<$0.001\\
 Years after pseudomonas infection & -0.44 & 0.04 & $<$0.001 & $<$0.001 & -0.44 & 0.04 & $<$0.001  & $<$0.001\\
 Age                               & -0.27 & 0.05 & $<$0.001 & $<$0.001 & -0.47 & 0.04 & $<$0.001  & $<$0.001\\
 Cohort 1948                       & 2.13  & 4.89 & 0.519    & 0.804    & 4.65  & 5.41 & 0.057     & 0.722\\
 Cohort 1958 			                 & -3.45 & 1.42 & 0.009    & 0.023    & -7.25 & 1.50 & $<$0.001  & $<$0.001\\
 Cohort 1978 			                 & 17.80 & 1.06 & $<$0.001 & $<$0.001 & 16.88 & 1.06 & $<$0.001  & $<$0.001\\
 Cohort 1988 			                 & 26.15 & 1.17 & $<$0.001 & $<$0.001 & 25.70 & 1.10 & $<$0.001  & $<$0.001\\
 Cohort 1998		                   & 29.17 & 1.83 & $<$0.001 & $<$0.001 & 28.47 & 1.63 & $<$0.001  & $<$0.001\\
 Pancreatic sufficiency            & 0.71  & 3.24 & 0.749    & 0.898    & 7.08  & 2.89 & $<$0.001  & 0.095 \\
 Age * Cohort 1948 	            	 & -0.08 & 0.14 & 0.444    & 0.703    & -0.03 & 0.13 & 0.553     & 0.996 \\
 Age * Cohort 1958		 				  	 & 0.08  & 0.06 & 0.125    & 0.235    & 0.28  & 0.06 & $<$0.001  & $<$0.001\\
 Age * Cohort 1978 		             & -0.80 & 0.07 & $<$0.001 & $<$0.001 & -0.72 & 0.06 & $<$0.001  & $<$0.001\\
 Age * Cohort 1988 		             & -0.76 & 0.11 & $<$0.001 & $<$0.001 & -0.77 & 0.09 & $<$0.001  & $<$0.001\\
 Age * Cohort 1998 		             & 0.44  & 0.43 & 0.244    & 0.376    & 0.43  & 0.37 & 0.152     & 0.356 \\
 Age * Pancreatic sufficiency      & 1.13  & 0.22 & $<$0.001 & $<$0.001 & 0.80  & 0.19 & $<$0.001  & $<$0.001\\
 \end{tabular}}}
 \caption{Estimates of the fixed effects for the Normal and NIG models. 
 Age is centered at 5. Cohort 1968, absence of diabetes, absence of pancreatic sufficiency are 
 the reference categories. p-lower and p-upper indicate bounds for p-values by taking into account Monte Carlo error 
 in the parameter and standard error estimates.}
 \label{tab:FEVpars}
 \end{table}

 \subsection{Progression towards end-stage renal failure}
 Our second application uses clinical data on the kidney function of 
 primary care patients from the northern English city of Salford
 who are in high-risk groups for
chronic kidney disease. The outcome variable is 
 eGFR (estimated Gromerular Filtration Rate, in mL/min per 1.73$\mbox{m}^2$ of body surface area), 
 a proxy measurement for the patient's renal function calculated as
 \begin{equation}
 \mbox{eGFR} = 175 \times \left(\frac{\mbox{SCr}}{88.4} \right)^{-1.154} \times 
 \mbox{age}^{-0.203} \times 0.742^{\mbox{I(female)}} \times 1.21^{\mbox{I(black)}},
 \end{equation}
 where $\mbox{SCr}$ stands for serum creatinine measured in $\mu$mol/L (Levey et al, 1999).
 
 The data, previously analysed by \cite{diggle2015}, contain a total of 
 392,870 measurements on 22,910 patients, for whom total follow-up time 
 ranged from zero (i.e.~only baseline data is available) to 10.0 years, 
 whilst the number of measurements of eGFR ranged from 1 to 305. 
 Amongst the 22,910 patients, 11,833 (51.7\%) were male. 
 Baseline ages ranged between 13.7 and 102.1 with a mean of 65.4. 

 Figure \ref{fig:eGFR} shows traces for eight patients, chosen to illustrate some 
 particularly challenging features of the data. The unusually high degree of 
 irregularity in the follow-up times reflects the fact that the data derive from 
routine clinical practice. In particular, some patients provided many repeated 
 measurements over a relatively short time-period, probably during episodes of 
 inter-current illness.  

 \begin{figure}[t]
 \centering
 \includegraphics[width = 0.6\linewidth]{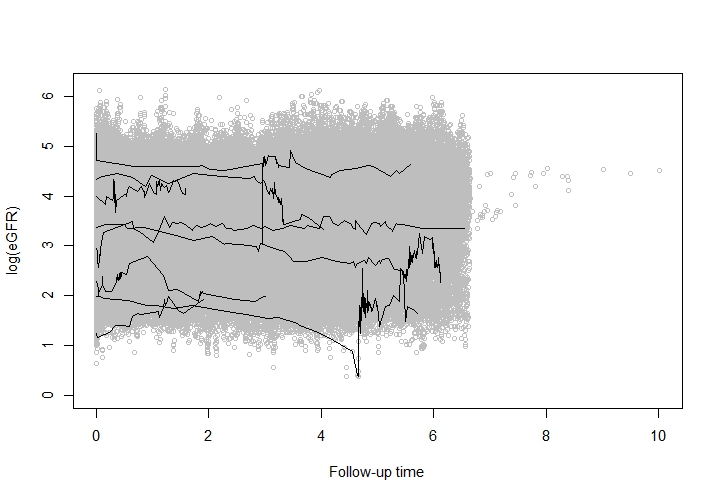}
 \caption{eGFR measurements (in log-scale) against follow-up time (in years) 
 in background as a scatter-plot. Data on eight patients are highlighted by 
 black lines connecting successive measurements.}
 \label{fig:eGFR}
 \end{figure}
    
 Clinical care guidelines in the UK include a recommendation that any person in 
 primary care who appears to be losing kidney function at a relative rate of at 
 least 5\% per year should be considered for referral to specialist secondary care. 
 Our primary objective in analysing these data is therefore to develop a method for 
 identifying, for each subject and in real-time, when this criterion is first met.
 
 As in \cite{diggle2015}, we use a log-transformed outcome variable, $Y = \log(\mbox{eGFR})$, and specify a
 model of the form 
 \begin{equation}
 Y_{ij} = \mv{x}_{ij}^{\top} \mv{\beta} + U_i + W_i(t_{ij}) + Z_{ij}.
 \label{eq:renal_mixed}
 \end{equation}
 In (\ref{eq:renal_mixed}), each $\mv{x}_{ij}$ includes 
 sex, baseline age, follow-up time ($t_{ij}$) and 
 a piece-wise linear function of age with a slope change at age 56.5.
 The processes $W_i(t) $ are integrated 
 random walks as in \cite{diggle2015}.  

 We first fit the model under Gaussian assumptions for the 
$U_i$, $W_i(\cdot)$ and $Z_{ij}$ components. The residual QQ plot
shown as Figure 5 of \cite{diggle2015} shows longer-than-Gaussian
tails. 
 As for the CF example, the source of this deviation
from the Gaussian assumptions is unknown. 
 Therefore, we proceed by assuming NIG distributions for 
each of the three stochastic components.  
 Estimates of $\nu$ based on this model are 99.93, 0.01 and 0.19 for 
 $U_i$, $W_i(\cdot)$ and $Z_{ij}$, respectively. 
 As the magnitude of $\hat{\nu}$ for $U_i$ indicates
close-to-Gaussian behaviour
 our final model assumes a Normal distribution for $U_i$, 
 and NIG distributions for $W_i(\cdot)$ and $Z_{ij}$.

 Figure \ref{fig:excursions} shows, 
 for two patients,
their observed data and the concurrent (``nowcasting'')
probabilities of 
 meeting the clinical guideline for referral to specialist
care. These are derived from the predictive distributions  
 $[Y_{ik}^*|Y_{i1}, \ldots, Y_{ik}]$, 
 where $Y_{ik}^* = Y_{ik} - Z_{ij}$; see \eqref{eq:renal_mixed}.  
 Results are shown for our preferred model, with Normally distributed 
 $U_i$, and NIG distributed
 $W_i(\cdot)$ and $Z_{ij}$ components, and 
for the corresponding Gaussian model. As would be expected, 
for each patient the general pattern of the predictive probabilities
is similar under both modelling assumptions, but there are some
substantial quantitative differences and the ranking of each
pair of predictive probabilities is not consistent.
 The two sets of model-based
predictions reflect different partitionings 
of the intra-patient variation
into signal and noise components, and the balance between the two
is affected in subtle ways
by the pattern of follow-up times and their associated
measurements.

\begin{figure}[t]
	\centering
	\includegraphics[width = 4in, height = 3in]{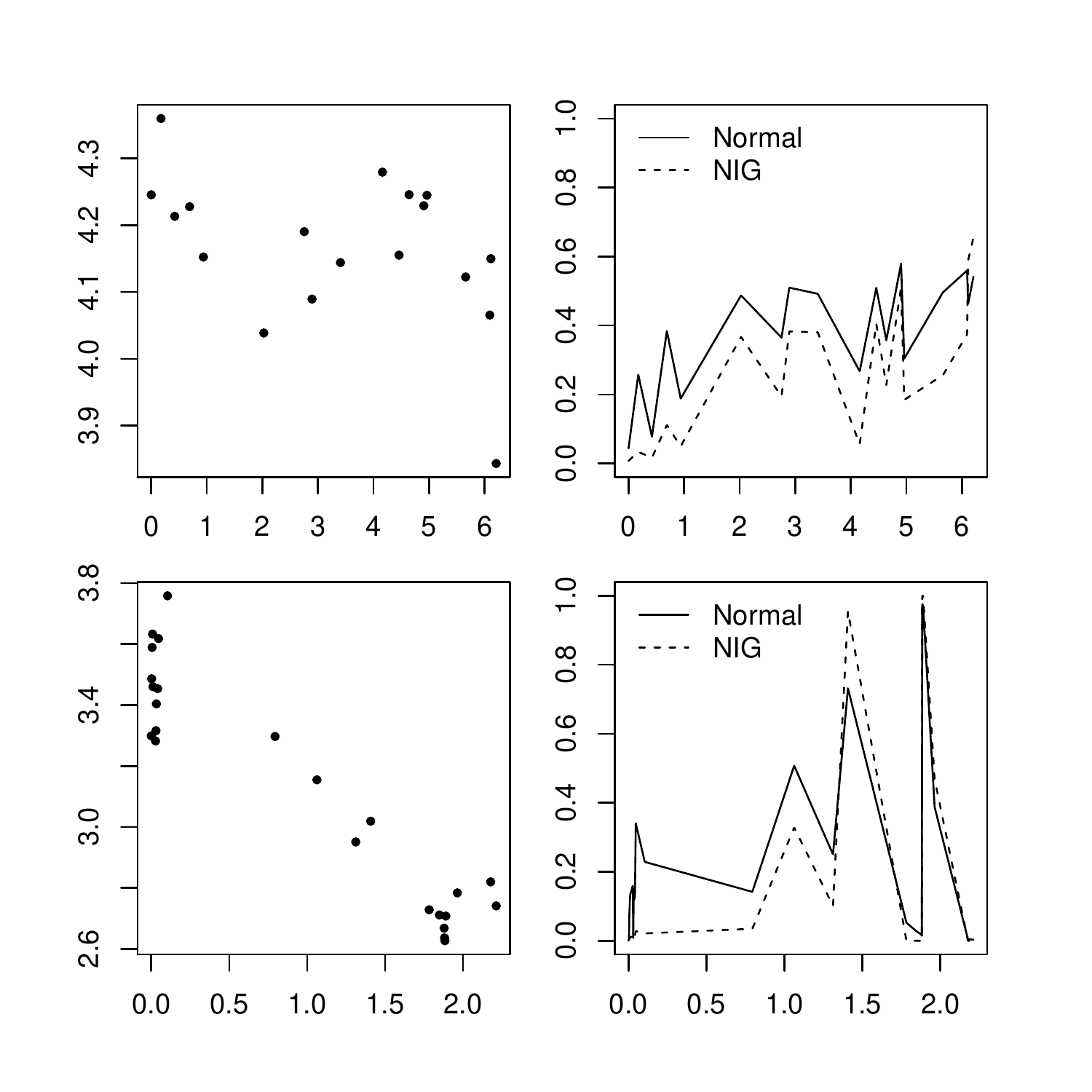}
	\caption{Left panel: follow-up time (in years) vs. log(eGFR) for two patients. 
	Right panel: probabilities of meeting the clinical guideline for the patients. }
	\label{fig:excursions}
\end{figure}

 \section{Software} 
 \label{sec:package}

 We have implemented the methodology presented in this paper 
 in the {\tt R} package {\tt ngme}. A development version of 
 the package  is available from {\tt https://bitbucket.org/davidbolin/ngme}. 
 The package includes functions for 
 parameter estimation and for subject-level prediction for the 
 class of models defined by \eqref{eq:three-component}, 
 with the following features.
 \begin{enumerate}
 \item Any linear model can be specified for the regression term, 
 $\mv{x}_{ij}^\top \mv{\beta}$, and for the subject-level random effect, 
 $\mv{d}_{ij}^\top \mv{U}_i$, using the standard {\tt R} model formula syntax.
 \item The random effects distribution can be chosen as Normal or NIG. 
 \item The covariance structure of the $W_i(t)$ can be specified as a 
 stationary, exponentially correlated  process or as a non-stationary integrated random walk, or 
 omitted altogether to give non-Gaussian versions
 of the Laird-Ware model.  The distribution for the process 
 can be specified as Normal, NIG, GAL or Cauchy. 
 \item The distribution of the measurement error terms can be specified as  Normal, NIG or t. 
 \item Subject-level predictions can be obtained either through nowcasting (conditioning on a subject's past  
 and current measurement data),
 smoothing (conditioning on all of a subject's data) or 
 forecasting (conditioning on all of a subject's past data).
 \item The generic {\tt R} functions, {\tt print}, {\tt summary}, {\tt plot}, 
 {\tt fitted} and {\tt residuals} are available for the main estimation and prediction functions.
 \item The renal data-set is included.
 \end{enumerate}

We plan to extend the package's functionality to a wider 
range of models for the stochastic process component $W_i(t)$,
including a general Mat\'{e}rn correlation strcuture.

 \section{Discussion} 
 \label{sec:discussion}

 The Gaussian version of the linear mixed model (\ref{eq:three-component}) 
 represents the standard approach to analysing real-valued repeated measurement data. 
 Typically, the simplified version without the Gaussian process term $W_i(t_{ij})$ 
 suffices when the number of follow-up times per subjects
 is small, whilst the version with the $W_i(t_{ij})$ term, often in conjunction with a
 simple random intercept $U_i$ in place of the general
 term $\mv{d}_{ij}^\top \mv{U}_i $, gives a better fit to data with long follow-up sequences.
 Concerns have often been raised about the legitimacy 
 of the Gaussian assumption, and in particular about the consequences of fitting Gaussian 
 models when elements of the underlying process have longer-than-Gaussian tails or skewness. 
 This has led to an extensive literature, which we reviewed in Section \ref{sec:literature}. 
 However, to the best of our knowledge the current paper is the first to provide a flexible 
 implementation in which departure from Gaussianity can be assessed independently for each of 
 the three stochastic components of (\ref{eq:three-component}). 
 
In our re-analysis of the cystic fibrosis data, 
 inferences on fixed effects showed only small changes when non-Gaussian
behaviour is taken into account. 
 Our re-analysis of the renal data also finds evidence of 
 non-Gaussian behaviour. which in this case 
 matters more, because it can change materially the 
 point at which individual patients in primary care are identified as meeting the accepted criterion 
 for referral to secondary care.

 We have emphasised the importance of building a computationally efficient algorithm for routine 
 maximisation of the likelihood. Arguably, computational efficiency is of secondary importance 
 in confirmatory analysis. Once the statistical analysis protocol is determined, it matters little 
 whether it takes minutes, hours or days of computing time to analyse a data-set that typically will 
 have taken weeks, months or years to collect. However, during the iterative model-building cycle 
 that characterises exploratory data analysis, the inability to fit and compare different models in 
 real-time is a severe impediment. 

 The applications described in Section \ref{sec:examples} show that the sub-sampling scheme 
 introduced in Section \ref{sec:subsampling} can perform very well. 
One topic of future research is a more thorough investigation of 
 how to optimise the sub-sampling.  
 Another is to develop graphical diagnostic tools for 
 non-Gaussian models, the need for which is
discussed in \cite{singer2017}.

Generalised linear mixed models provide a framework
 for handling non-Gaussian sampling distributions. 
This form of non-Gaussian behaviour is complementary to the
kind of non-Gaussian process behaviour that we have addressed in this article. A natural extension to our proposed models
would be to generalised linear mixed models
for binary or count data with non-Gaussian random effects.
 However, non-Gaussian behaviour will naturally be harder to detect from count or binary data than from measurement data. Binary data in particular can  be considered as a heavily censored version of measurement data. For example, a logistic regression model can be interpreted as a linear regression model for a real-valued response $Y$ in which only the sign of $Y$ is observed.

 Clinical repeated measurement data are often coupled with 
time-to-event outcomes, e.g. death. So called {\it joint models} for
repeated measurement and time-to-event outcomes have been widely studied; for a recent book-length account, see \cite{rizopoulos2012}.  However, essentially all of this literature assumes that 
any random effect components are Gaussian.
 A natural way of 
extending the methodology presented in this paper to 
 joint modelling problems, by analogy with much of the current literature on Gaussian joint models,
would be to combine the linear mixed model  \eqref{eq:three-component} with a log-linear Cox process 
model for the time-to-event outcome, in which the stochastic process $W_i(t)$ in the repeated measurement 
sub-model is correlated with a second stochastic process, $W_i^*(t)$ say, such that
$\exp\{W_i^*(t)\}$ constitutes a time-dependent frailty for the $i$th subject. 

Another possible extension of the methodology presented in this paper would be to multivariate settings, in which
more than one repeated biomarker measurement is collected 
for each patient, sometimes with different follow-up schedules for different biomarkers. 



 
 








\begin{appendix}
\section*{Appendices}

\section{Sampling the variance components}
The conditional
distributions of the random effects $\mv{U}_i$ and 
stochastic processes $\mv{W}_i$, given the variance components and the 
data $\mv{Y}_i$, are Normal, as stated in section \ref{sec:HierRep}. Thus sampling these variables given the variance components is straightforward. For the Gibbs sampler  used within the stochastic gradient
algorithm (Section \ref{sec:stochastic_gradient})
we also need the distribution of the variance components conditional 
on $\mv{U}_i,\mv{Y}_i$ and $\mv{W}_i$. If the distributions of the variance components are GIG then so are the conditional distribution. For brevity we only show here the conditional distribution when
\begin{align*}
 V_{ij}^Z  &\sim {\rm IG}(\nu^Z, \nu^Z), \\
 V_i^U  &\sim {\rm IG}(\nu^U, \nu^U), \\
 V_{ij}^W  &\sim {\rm IG}(h^2_{ij}\nu^W, \nu^W).
\end{align*}
The above distributional assumptions
 imply that {\it a priori} 
the noise and random effect components have NIG distributions
 and the process $\mv{W}_i$ is a NIG processes.  The resulting conditional distributions are:
\begin{align*}
V^U_i | \mv{U}_i, \mv{\Theta}&\sim GIG\left(-\frac{d^U+1}{2}, \nu^U + \left(\mv{\mu}^U \right)^{\top} \mv{\Sigma}^{-1}\mv{\mu}^U  ,\nu^U + \left(\mv{U}_i + \mv{\mu}^U \right)^{\top} \mv{\Sigma}^{-1} \left(\mv{U}_i + \mv{\mu}^U \right)\right), \\
\mv{V}^Z_{i} |\mv{U}_i,\mv{W}_i,Y_{i}, \mv{\Theta}&\sim GIG\left(-\mv{1} ,\mv{1} \cdot \nu^Z , \nu^Z+ \frac{\left(Y_{i} - \mv{x}_{i}\mv{\beta} + \mv{d}_{i} \mv{U}_i - \mv{A}_{ij} \mv{W}_i\right)^2}{\sigma^2}\right), \\
\mv{V}^W_i | \mv{W}_i, \mv{\Theta}&\sim GIG\left(-\mv{1},\mv{1} \cdot \left( \nu^W + (\mu^W)^2 \right), \left(\mv{K}\mv{W}_i  + \mv{h} \mu^w \right)^2 + \nu^W\mv{h}^2\right),
\end{align*}
where $d^U$ is the dimension of $\mv{U}_i$.
In the second and third line above,
 the individual elements of $\mv{V}_i^Z$ and $\mv{V}_i^W$
are independent with distributions ${\rm GIG}(p_i,a_i,b_i)$, and 
the $\cdot^2$ operation is applied element-wise. This can easily be generalized to the full GIG family.

\section{Gradients and Observed Fisher Information}
In this section, we present the gradients and Hessians required for the parameter estimation. To do this, we need
some results from matrix calculus. 
The \emph{vec}-operator transforms a matrix into a vector by stacking its columns.  
The \emph{vech} operator also transforms an $n\times n$ 
matrix into a vector but  removes all the subdiagonal elements. 
Finally, the duplication matrix, $\mv{D}_n$, 
is such that for any symmetric matrix $\mv{A}$,
$
\mv{D}_n ve{\rm CH}(\mv{A}) = vec(\mv{A}). 
$ 
For a detailed description of these redsults see \cite{magnus2007}. 

\subsection{Random effects parameters}
For the random effects parameters we derive the gradient and the Observed Fisher information matrix for $\mv{\Sigma},\mv{\mu}$, and $\mv{\beta}$. Here, we write the random effect $U$
conditional on the random variance $V$ as
$$
\mv{U} = - {\rm E}[V^U] \mv{\mu} + \mv{\mu}V^U + \mv{\Sigma}^{1/2}\mv{Z},
$$
where $\mv{Z} \sim {\rm N}(\mv{0},\mv{I}_d)$. This deviates slightly from the form used in the main bdy of
the paper, where we assumed that
 ${\rm E}[V^U]=1$. Here, we add
the term $-{\rm E}[V^U]\mu$ to ensure that $\mv{U}$ has zero expectation. If the expectation of $V^U$ is unbounded this is not possible, and we would either drop the term
or replace it with the mode. To simplify the presentation,  we 
here assume that the expectation exists. 

The relevant part of the likelihood
for the random effect parameters is 
$$
 \mv{L}(\mv{\Theta}; \mv{y}, \mv{U}, \mv{V}) \propto|\mv{\Sigma}|^{-N/2}  \exp \left(-\sum_{i=1}^N \frac{1}{2V^U_i}\left(  \mv{U}_i + {\rm E}[V^U_i]\mv{\mu} - V^U_i \mv{\mu} \right)^{T}\mv{\Sigma}^{-1} \left( \mv{U}_i + {\rm E}[V^U_i]\mv{\mu} - V^U_i \mv{\mu} \right) \right)
$$

The gradient vector and observed information matrix for the variance matrix  $\mv{\Sigma}$
 can be derived  as
\begin{align*}
\nabla_{ve{\rm CH}(\mv{\Sigma})} \log \mv{L}(\mv{\Theta}; \mv{y}, \mv{U}, \mv{V}) &= \frac{1}{2} \mv{D}_d^{\top} \left( \mv{\Sigma}^{-1} \otimes \mv{\Sigma}^{-1} \right) vec( \sum_{i=1}^N  \frac{\mv{M}_i}{V^U_i} -   N\mv{\Sigma}), \\
E[ \nabla_{ve{\rm CH}(\mv{\Sigma})} \nabla_{ve{\rm CH}(\mv{\Sigma})} \log \mv{L}(\mv{\Theta}; \mv{y}, \mv{U}, \mv{V}) ] &= N
\mv{D}_d^{\top} \left( \mv{\Sigma}^{-1} \otimes \mv{\Sigma}^{-1} \right) \mv{D}_d.
\end{align*}
where $\mv{M}_i  = \left( \mv{U}_i +{\rm E}[V^U_i] \mv{\mu} - V^U_i \mv{\mu} \right) \left( \mv{U}_i +{\rm E}[V^U_i] \mv{\mu} - V^U_i \mv{\mu} \right)^{\top}$ and $\mv{D}_d$ is the duplication matrix
 \citep[p.389-390]{magnus2007}.
The gradient and the observed information matrix for $\mu$  are
\begin{align*}
\nabla_{\mv{\mu}} \log \mv{L}(\mv{\Theta}; \mv{y}, \mv{U}, \mv{V})  &=  \sum_{i=1}^N \frac{(-{\rm E}[V^U_i]+ V^U_i)}{V^U_i} \mv{\Sigma}^{-1} \left( \mv{U}_i + {\rm E}[V^U_i]\mv{\mu} - V^U_i \mv{\mu} \right),   \\
E[ \nabla_{\mv{\mu}} \nabla_{\mv{\mu}} \log \mv{L}(\mv{\Theta}; \mv{y}, \mv{U}, \mv{V})  ] &=(\sum_{i=1}^N {\rm E}[\frac{(V^U_i -{\rm E}[V^U_i])^2}{V^U_i}])  \mv{\Sigma}^{-1}  \\
&=   \mv{\Sigma}^{-1}  \sum_{i=1}^N - {\rm E}[V^U_i]^2  {\rm E}[1/V_i^U] - {\rm E}[V^U_i].   
\end{align*}
We can only compute these if ${\rm E}[1/V_i^U] $ and ${\rm E}[V^U_i]$ are bounded. This excludes the Gamma 
and inverse Gamma distributions for $V$, for which
 ${\rm E}[\frac1{V_i^U}]$ and ${\rm E}[V^U_i]$, respectively,
are unbounded
in parts of the parameter space. 

\subsection{Regression and noise parameters}
We now consider
the fixed effect parameters $\beta$ and the measurement
noise  variance $\sigma^2$. The relevant part of the likelihood for 
these parameters is
$$
\mv{L}(\mv{\Theta}; \mv{y}, \mv{U}, \mv{V}) \propto  \sigma^{- \sum_{i=1}^N m_i}  \exp \left(-\sum_{i=1}^N \frac{1}{2\sigma^2} \mv{e}^{\top}_{i}\mv{e}_{i} \right).
$$
where $\mv{e}_{i} = \mv{y}_{i} - \mv{x}_i\mv{\beta} - \mv{d}_i\mv{U}_i - \mv{A}_i \mv{W}_i$. 
The gradient and observed information of  the measurement noise standard deviation, $\sigma$, is 
\begin{align*}
\nabla_{\sigma} \log \mv{L}(\mv{\Theta}; \mv{y}, \mv{U}, \mv{V}) &= - \sum_{i=1}^N  \frac{m_i}{\sigma}  + \frac{1}{\sigma^3} \sum_{i=1}^N  \left(\mv{e}_i \cdot \frac{1}{\mv{V}^Z} \right)^{\top} \mv{e}_i, \\
\E [\nabla_{\sigma} \nabla_{\sigma}  \log \mv{L}(\mv{\Theta}; \mv{y}, \mv{U}, \mv{V})] &= \sum_{i=1}^N \frac{m_i}{\sigma^2} - \frac{3}{\sigma^4} \sum_{j=1}^{m_j} \E \left[\frac{e_{ij}^2}{V^Z_{ij}}\right] = -2\sum_{i=1}^N \frac{m_i}{\sigma^2}.
\end{align*}
For the fixed effect the gradient and the observation matrix is

\begin{align*}
\nabla_{\mv{\beta}} \log \mv{L}(\mv{\Theta}; \mv{y}, \mv{U}, \mv{V}) &= - \frac{1}{\sigma^2} \sum_{i=1}^N  \mv{x}_i  \cdot \left( \frac{1}{\mv{V}^Z} \right)^{\top} \mv{e}_i, \\
\E [\nabla_\mv{\beta} \nabla_{\mv{\beta}}  \log \mv{L}(\mv{\Theta}; \mv{y}, \mv{U}, \mv{V})] &=- \frac{1}{\sigma^2} \sum_{i=1}^{N} \E \left[\frac{1}{V^Z}\right]\mv{x}_i\mv{x}^{\top}.
\end{align*}
Note that the computation of the gradient for the fixed effect
parameters can be moved to the random effect part. This choice is important since it affets the variance of the gradient. 
This choice is the same as that between central and non-central parametrizations in  MCMC; see \cite{Papas2007}. In the package {\tt ngme}, both choices are used, and the two gradients are weighted according to their observed information matrices.
\subsection{Process parameters}
The relevant part of the likelihood for the process
parameters is
$$
\mv{L}(\mv{\Theta}; \mv{y}, \mv{U}, \mv{V}) \propto|\mv{K}|^{N}  \exp \left(-  \frac{1}{2}\sum_{i=1}^N\left(  \mv{E}_i + \mv{h}\mu -  \mv{V}^W_i\mu  \right)^{T} diag(\mv{V}_i^W)^{-1} \left( \mv{E}_i + \mv{h}\mu - \mv{V}^W_i \mu \right) \right),
$$
where $\mv{E}_i = \mv{KW}_i$. We start by presenting the general form for differentiation of the likelihood with respect to a generic operator parameter, $\theta$. Since the results follow from standard matrix calculus, we omit the details of the computations.
We define the matrix $\mv{K}_\theta$ to have elements
 $(\mv{K}_\theta)_{ij} = \frac{dK_{ij}}{d\theta}$. The gradient is
then given by
\begin{align*}
\nabla_{\theta} \log \mv{L}(\mv{\Theta}; \mv{y}, \mv{U}, \mv{V}) =&  N {\bf tr}\left(\mv{K}_\theta \mv{K}^{-1}\right) - \sum_{i=1}^N \mv{W}_i diag(\mv{V}_i^W)^{-1} \left( \mv{E}_i + \mv{h}\mu - \mv{V}^W_i \mu \right).  
\end{align*}
The cost of computing the observation matrix is prohibitive,
so we omit it.
For the shift parameter, the gradient and the observed fisher information matrix are given by
\begin{align*}
\nabla_{\mv{\mu}} \log \mv{L}(\mv{\Theta}; \mv{y}, \mv{U}, \mv{V})  &=  \sum_{i=1}^N\left(   \mv{h} -  \mv{V}^W_i \right)^{T} diag(\mv{V}_i^W)^{-1} \left( \mv{E}_i + \mv{h}\mu - \mv{V}^W_i \mu \right) , \\
E[ \nabla_{\mv{\mu}} \nabla_{\mv{\mu}} \log \mv{L}(\mv{\Theta}; \mv{y}, \mv{U}, \mv{V})  ] &=\sum_{i=1}^N \sum_{j=1}^n {\rm E}[\frac{(\mv{V}^W_{i,j} -\mv{h}_j)^2}{\mv{V}^U_{i,j}}]   \\
&=   \sum_{i=1}^N \sum_{j=1}^n  -\mv{h}_j^2  {\rm E}[\frac1{\mv{V}_{ij}^U}] - \mv{h}_j. 
\end{align*}
\subsection{Variance parameters}
because we sample all variances $V$ 
associated with the stochastic components in the model, the gradients for the parameters of the
different variances
only depend on the form of the distribution for each specific parameter. For this reason, we only present the result for the process
 variances. Results for the other variance parsameters are 
derived in the same way.
For the NIG process, we only have the parameter $\nu^Z$. The relevant part of the likelihood is

$$
\mv{L}(\mv{\Theta}; \mv{y}, \mv{U}, \mv{V}) \propto \sum_{i=1}^N \nu^{n/2} \exp \left(-0.5  \nu\mv{h}^{\top}_i \left( \frac{\mv{h}_i}{\mv{V}^W_i}\right) - 0.5 \nu \mv{1}^{\top} \mv{V}^W_i + \nu \mv{1}^{\top}\mv{h}_i  \right).
$$
Thus the gradient and the observed Fisher information are
\begin{align*}
\nabla_{\nu} \log \mv{L}(\mv{\Theta}; \mv{y}, \mv{U}, \mv{V}) &=  \sum_{i=1}^N  \frac{n}{2\nu} -0.5  \mv{h}^{\top}_i \left( \frac{\mv{h}_i}{\mv{V}^W_i}\right) - 0.5 \mv{1}^{\top} \mv{V}^W_i + \mv{1}^{\top} \mv{h}_i  \\
{\rm E}\left[\nabla_{\nu} \log \mv{L}(\mv{\Theta}; \mv{y}, \mv{U}, \mv{V}) \right]&=  -\sum_{i=1}^N  \frac{n}{2\nu^2}.
\end{align*}

\end{appendix}

\end{document}